\documentclass[sigconf]{acmart}
\AtBeginDocument{%
  \providecommand\BibTeX{{%
    \normalfont B\kern-0.5em{\scshape i\kern-0.25em b}\kern-0.8em\TeX}}}

\setcopyright{acmcopyright}
\copyrightyear{2022}
\acmYear{2022}
\acmDOI{XXXXXXX.XXXXXXX}

\acmConference[SIGSPATIAL ’22 November 1-4, 2022]{November 1-4, 2022}{Seattle, USA}

\acmPrice{15.00}
\acmISBN{978-1-4503-XXXX-X/18/06}

\usepackage{amsthm}
\usepackage{mathtools}
\usepackage{caption}
\usepackage{subcaption}
\usepackage{enumitem}
\usepackage{algorithm}
\usepackage{multicol}
\renewenvironment{proof}{{\bfseries Proof.}}{\qed}

\usepackage[noend]{algpseudocode}
\usepackage{amsmath}
\algnewcommand\algorithmicforeach{\textbf{for each:}}
\algnewcommand\ForEach{\item[ \algorithmicforeach]}
\algdef{S}[FOR]{ForEach}[1]{\algorithmicforeach\ #1\ \algorithmicdo}
\algtext*{EndWhile}
\algtext*{EndIf}

\AtBeginDocument{%
  \providecommand\BibTeX{{%
    \normalfont B\kern-0.5em{\scshape i\kern-0.25em b}\kern-0.8em\TeX}}}

\begin{document}

\title{Towards a Tighter Bound on Possible-Rendezvous Area}

\title{Towards a Tighter Bound on Possible-Rendezvous Areas: Preliminary Results}

\author{Arun Sharma}
\email{{sharm485@umn.edu}}
\affiliation{%
  \institution{University of Minnesota, Twin Cities}
  \city{Minneapolis}
  \state{Minnesota}
  \country{USA}
  \postcode{43017-6221}
}

\author{Jayant Gupta}
\email{gupta423@umn.edu}
\affiliation{%
  \institution{University of Minnesota, Twin Cities}
  \city{Minneapolis}
  \state{Minnesota}
  \country{USA}
  \postcode{43017-6221}
}

\author{Subhankar Ghosh}
\email{ghosh117@umn.edu}
\affiliation{%
  \institution{University of Minnesota, Twin Cities}
  \city{Minneapolis}
  \state{Minnesota}
  \country{USA}
  \postcode{43017-6221}
}


\renewcommand{\shortauthors}{Arun Sharma, Jayant Gupta and Subhankar Ghosh}

\begin{abstract}
Given trajectories with gaps, we investigate methods to tighten spatial bounds on areas (e.g., nodes in a spatial network) where possible rendezvous activity could have occurred. The problem is important for reducing the onerous amount of manual effort to post-process possible rendezvous areas using satellite imagery and has many societal applications to improve public safety, security, and health. The problem of rendezvous detection is challenging due to the difficulty of interpreting missing data within a trajectory gap and the very high cost of detecting gaps in such a large volume of location data. Most recent literature presents formal models, namely space-time prism, to track an object's rendezvous patterns within trajectory gaps on a spatial network. However, the bounds derived from the space-time prism are rather loose, resulting in unnecessarily extensive post-processing manual effort. To address these limitations, we propose a Time Slicing-based Gap-Aware Rendezvous Detection (TGARD) algorithm to tighten the spatial bounds in spatial networks. We propose a Dual Convergence TGARD (DC-TGARD) algorithm to improve computational efficiency using a bi-directional pruning approach. Theoretical results show the proposed spatial bounds on the area of possible rendezvous are tighter than that from related work (space-time prism). Experimental results on synthetic and real-world spatial networks (e.g., road networks) show that the proposed DC-TGARD is more scalable than the TGARD algorithm.

\end{abstract}




\maketitle

\section{Introduction}
Given trajectories with gaps (i.e., missing data), we investigate methods to tighten bounds on the spatial networks (e.g., road network, river network, etc.) for detecting potential rendezvous or meetup locations. Figure \ref{fig:ProblemStatmeent} (a) shows a pair of trajectory gaps on an underlying spatial network topology where a gap exists from $t=2$ to $t=6$ for both trajectory $1$ (blue) and trajectory $2$ (red) with a given object speed at $1$ unit/second. Current approaches output possible meeting locations of two objects via the intersection of two space-time prisms. For instance, Figure \ref{fig:ProblemStatmeent} (b) shows six interpolated nodes (green) that qualify as rendezvous locations for both objects. Figure \ref{fig:ProblemStatmeent}(c) shows the result after a time slicing method has reduced the number of nodes (by a factor of $3$), filtering out nodes $N_{10}$, $N_{17}$, $N_{12}$, and $N_{19}$ (yellow) and leaving $N_{11}$ and $N_{18}$ (green) as the output. The resultant interpolated nodes are then sent to human analysts for further investigation.
\begin{figure*}[ht]
    \centering
    \includegraphics[width=1.0\textwidth]{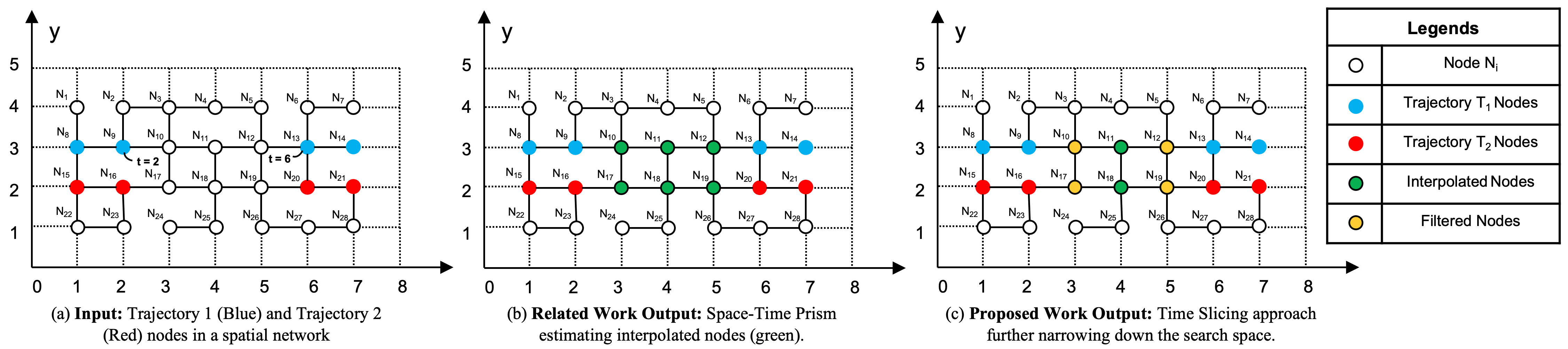}
    \captionsetup{justification=centering}
    \caption{A illustration of Rendezvous Region Detection in Spatial Networks (Best in color)}
    \label{fig:ProblemStatmeent}
\end{figure*}

Reducing the size of the space with possible rendezvous nodes is important for helping human analysts to detect and analyze trajectory data gaps. The smaller spatial area can be more effectively verified with ground truth via satellite imagery which can further aid early-stage decision-making. The problem has many societal applications related to homeland security, public health and safety etc. For instance, maritime safety involves monitoring activities such as illegal fishing and illegal oil transfers and transshipments \cite{transhipment}. Similarly, public health officials can analyze a given affected area where two objects could have met to control the potential spread of the disease. This paper focuses on the specific use case of improving public safety when two objects intentionally mask their movements to meet secretly within a trajectory gap area.

The problem is challenging since it is hard to model and interpret specific behavioral patterns, especially the rendezvous of two or more objects within a trajectory gap. Many methods rely on linear interpolation, which may lead to many missed patterns since moving objects do not always travel in a straight path. Methods based on the space-time prism are more geometrically accurate but identify large spatial regions, resulting in a time-intensive operation for the post-processing step performed by the human analyst. Further, such methods have a high computation cost due to large data volume. This work proposes computationally efficient time-slicing methods to effectively capture rendezvous regions with tighter spatial bounds in spatial networks.

The traditional literature \cite{zheng2015trajectory,demiryurek2011online} on the mobility patterns of objects in spatial networks considers realistic scenarios and events (e.g., traffic congestion) and other behavioral patterns \cite{dodge2008towards}. However, little attention has been given to movement patterns within trajectory gaps . Most works in this area are limited to linear interpolation methods using shortest path discovery \cite{ding2008efficient,yuan2010interactive}. Other works consider an object’s motion uncertainty via geometric-based methods (e.g., space-time prisms \cite{miller1991modelling,kuijpers2009modeling,kuijpers2011analytic}) using spatial geo-ellipse boundaries constructed via motion parameters (e.g., speed). One recent work \cite{uddin2017assembly} does explicitly consider rendezvous or meetup queries in a spatial network, which the authors call \textit{assembly queries}, but they used a loosely bounded geo-ellipse estimation from the space-time prism model. Our work proposes time-slicing methods to give tighter bounds on a given rendezvous region and provide further computational speedup using a dual convergence approach.

\textbf{Contributions:} The paper contributions are as follows:
\begin{itemize}[noitemsep,topsep=0pt]
\item We propose a time slicing model and theoretically show that it provides a tighter bound on possible rendezvous area relative to traditional space-time prisms.
\item Using the time slicing model, we propose a Time Slicing-based Gap-Aware Rendezvous Detection (TGARD) algorithm to effectively detect rendezvous nodes and a Dual Convergence TGARD (DC-TGARD) algorithm using bi-directional pruning to improve the computational efficiency.
\item We provide a theoretical evaluation for both algorithms based on correctness, completeness, and time complexity.
\item We validate both algorithms experimentally based on solution quality and computation efficiency on both synthetic and real-world datasets.
\end{itemize}

\textbf{Scope:} The paper proposes TGARD and DC-TGARD algorithms to tighten the spatial bounds of the rendezvous region in the spatial networks. We do not consider gaps with short time intervals (i.e., minutes, seconds etc.) or spatial areas that are low density or sparse. Kinetic prisms \cite{kuijpers2017kinetic} fall outside the scope of the paper. In addition, we do not model the rendezvous of the object's in trajectories without gaps. The proposed framework has multiple phases (i.e., filter, refinement, and calibration), but we limit this work to the filter phase. The refinement phase requires input from a human analyst and is not addressed here, and calibration of the cost model parameters is outside the scope.

\textbf{Organization:} The paper is organized as follows: Section \ref{sec:Problem Definition} introduces basic concepts, framework and the problem statement. Section \ref{sec:TimeSlicing} provides an overview of time slicing model. Section \ref{sec:Proposed Approach} describes the baseline TGARD and refined DC-TGARD algorithms. Section \ref{sec:TheoreticalEvaluation} shows theoretical evaluations for proposed algorithms based on correctness, completeness and asymptotic complexity. Experimental evaluations are presented for both algorithms and related work are presented in Section \ref{sec:ExperimentalEvaluation}. A broad and detailed literature survey is given in Section \ref{section:Related_work}. Finally, Section \ref{section:conclusion_future_work} concludes this work and briefly lists the future work.

\section{Problem Formulation}
\label{sec:Problem Definition}
\subsection{Framework}
\label{sec:framework}
We aim to identify possible rendezvous locations in a given set of trajectory gaps through a two-phase \emph{Filter} and \emph {Refine} approach. We introduce an intermediate time slicing filter to reduce number of interpolated nodes residing within the rendezvous region via proposed algorithms in a computationally efficient manner. The refinement phase further improves the solution quality so that human analysts may extract and analyze a comparatively fewer number of nodes involved in possible rendezvous by two (or more) objects. The inspection is further verified via satellite imagery to derive a possible hypothesis about the rendezvous activity (as shown in Figure \ref{fig:Framework}).

\begin{figure}[ht]
    \centering
    \includegraphics[width=0.47\textwidth]{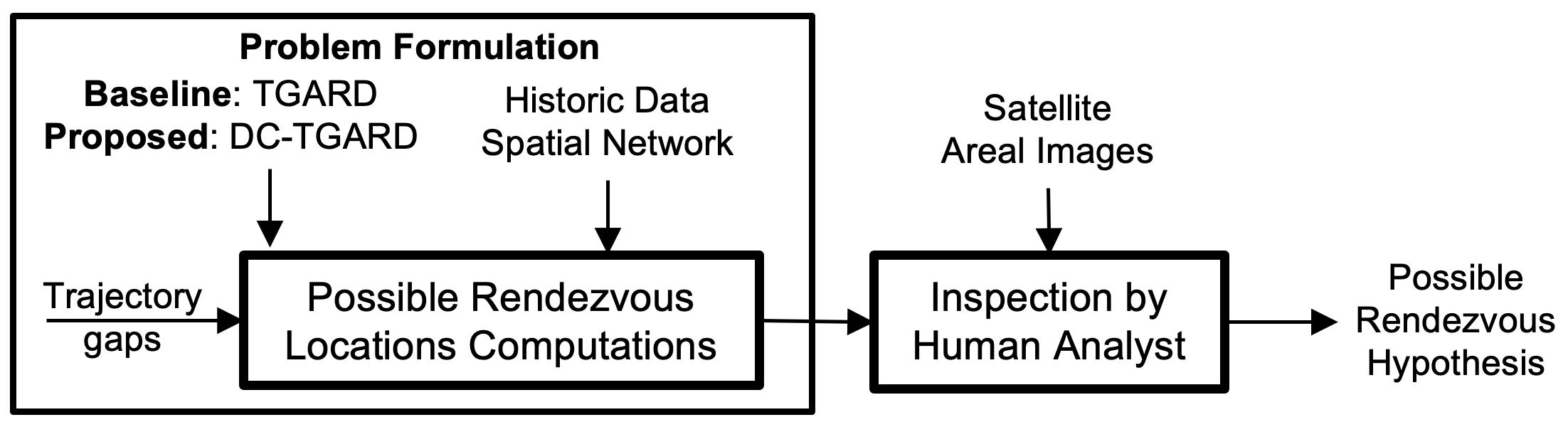}
    \caption{Framework for detecting possible rendezvous locations}
    \label{fig:Framework}
\end{figure}

\subsection{Basic Concepts}
A \textbf{spatial network} $G = (\mathcal{N}, \mathcal{E})$ consists of a node-set $\mathcal{N}$ and an edge set $\mathcal{E}$, where each element N $\in$ $\mathcal{N}$ is a geo-referenced point, while each element $E$ $\in$ $\mathcal{E}$ $\subseteq$ $N\times N$ has an edge weight $E_{w}$ i.e., the minimum time to travel from node $N_{i}$ to node $N_{j}$. Figure \ref{fig:ProblemStatmeent} shows a spatial network where circles represent nodes (e.g., $N_{1}$) and the lines represent edges. A road system is an example of a spatial network where nodes are intersections, and edges are segments.


A \textbf{trajectory gap} ($G_{i}$) is a spatial interpolated region within a missing location signal time period or Effective Missing Period (EMP) above a certain threshold $\theta$ (e.g., 30 mins) between two consecutive points. Such interpolated region can be modeled as a geo-ellipse on the x-y plane, which spatially delimits the extent of a moving object’s mobility given a maximum speed (MS) and effective missing signal period (EMP) (e.g., 30 mins etc) \cite{miller1991modelling}. Figure \ref{fig:TimeSliceSingle1} shows trajectory gap as spatial interpolated region in the form of a geo-ellipse (in blue) with $C_{1}$ and $C_{2}$ as focii and $t_{e}$-$t_{s}$ ($\geq$ 0) as $EMP$ with $MS_{i}$ as maximum speed.

A \textbf{possible rendezvous region} is defined as a set of overlapping nodes and edges when two trajectory gaps (i.e., spatially estimated in the form of geo-ellipse) involving different objects intersect with each other. Hence, such nodes can be defined as \textbf{possible rendezvous locations} for a set of objects within a trajectory gap. 
For instance, Figure \ref{fig:ProblemStatmeent} (b) shows interpolated nodes (in green) which may involved in a possible rendezvous location derived from the intersection of the two ellipses (as shown in Figure \ref{fig:TimeSliceDouble}) using a space-time prism based method \cite{uddin2017assembly} from trajectory $T_{1}$ and $T_{2}$. 

\subsection{Problem Formulation}\label{problem_definition}
The problem to optimally identify a trajectory gap region in a spatio-temporal domain is formulated as follows:\\
\textbf{Input:}
\begin{enumerate}
	\item A Spatial Network
	\item A set of $|N|$ Trajectory Gaps
	\item Historic Traffic Data.
\end{enumerate}
\textbf{Output:} A more tightly bound Possible Rendezvous Region\\
\textbf{Objective:} Solution Quality and Computational Efficiency\\
\textbf{Constraints:}
(1) Trajectories have long gaps (2) Maximum Acceleration is not available (3) Correctness and Completeness \\
Figure~\ref{fig:ProblemStatmeent} (a) shows the input as a two-dimensional representation of a trajectory gap. Figure~\ref{fig:ProblemStatmeent} (b) shows the output based on intersection of geo-ellipses resulting in interpolated nodes in green. Figure~\ref{fig:ProblemStatmeent} (c) provides a more refined output resulting in a smaller number of interpolated nodes for human analysts to inspect.

\section{Time Slicing Model}
\label{sec:TimeSlicing}
Our time slicing model uses a space-time prisms \cite{miller1991modelling} to provide a detailed representation of an object's physical space. We first describe space-time prisms and then describe the time slicing model.

Space-Time (ST) Prisms \cite{miller1991modelling} are a collection of spatial points bounded by a physical space which is represented as an interpolated region where moving objects could have passed at a given maximum speed MS. Figure \ref{fig:TimeSliceSingle1} shows the ellipse region in blue for a given time range [$t_{s}$,$t_{e}$] where $t_{s}$ denote the start time and $t_{e}$ denote the end time of the trajectory gap $G_{i}$. Equation \ref{ellipse} defines the geo-ellipse with focii ($x_{s},y_{s}$) and ($x_{e},y_{e}$) for a missing period $t_{e}$-$t_{s}$ as follows:
\begin{equation}
    \label{ellipse}
    \sqrt{(x - x_{s})^{2} + (y - y_{s})^{2}} + \sqrt{(x - x_{e})^{2} + (y - y_{e})^{2}} \leq (t_e - t_s)\times MS_{1}
\end{equation}

where, ($x_{s}$,$y_{s}$) and ($x_{e}$,$y_{e}$) are the start and end points of the trajectory gap with start time $t_{s}$ and end time $t_{e}$ ($t_{e}$ > $t_{s}$). The ellipse spatially delimits the extent of a moving object's mobility with maximum speed $MS_{1}$. 

Figure \ref{fig:TimeSliceDouble} shows the possible rendezvous region where we perform the spatial intersection of two ellipses $Ellipse_{1}$ $\in$ $G_{1}$ and $Ellipse_{2}$ $\in$ $G_{2}$ and time range is the calculated via inequality \ref{timerange_rendezvous}, 

\begin{equation}
    \label{timerange_rendezvous}
    [t_{s}^{Ellipse_{1}},t_{e}^{Ellipse_{1}}] \cap [t_{s}^{Ellipse_{2}},t_{e}^{Ellipse_{2}}] \neq \emptyset
\end{equation}

\begin{figure}[ht]
    \centering
    \includegraphics[width=0.30\textwidth]{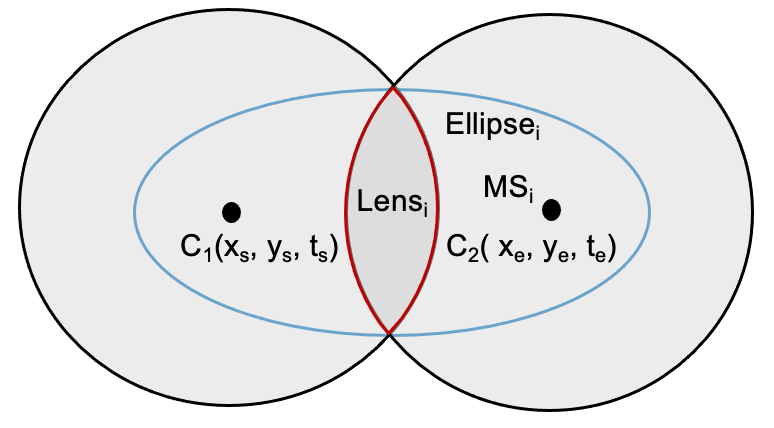}
    \caption{Time Slicing Model with Lens at time instant t}
    \label{fig:TimeSliceSingle1}
\end{figure}

A \textbf{Time Slice} is an object's physical space sampled within a geo-ellipse (i.e., a ST-Prism) of a trajectory gap. A time slice bounds the spatial intersection of two circles $C_{1}$ and $C_{2}$ at a given time instant $t$ where $t_{s}$ $\leq$ t $\leq$ $t_{e}$. Equations \ref{ineq2} and \ref{ineq3} represent the two circles generated from ($x_{s},y_{s}$) and ($x_{e},y_{e}$) at given time instant t (where $t_s$ $\leq t \leq$ $t_{e}$) as follows:
\begin{equation}
    \label{ineq2}
    C_{1} : (x_{i} - x_{s})^{2} + (y_{i} - y_{s})^{2} \leq (t - t_{s})^{2}MS_{i}^{2}
\end{equation}
 \begin{equation}
    \label{ineq3}
    C_{2} : (x_{i} - x_{e})^{2} + (y_{i} - y_{e})^{2} \leq (t_{e} - t)^{2}MS_{i}^{2}
\end{equation}
Figure \ref{fig:TimeSliceSingle1} shows the spatial bounds of a time slice denoted as \textbf{$Lens$} which is defined as the geometry derived via $C_{1}$ $\cap$ $C_{2}$ at a given time instant $t$. For instance, $Lens_{i}$ within an ellipse $Ellipse_{i}$ generated via maximum speed $MS_{i}$ at time $t$ $\in$ [$t_{s}$,$t_{e}$] i.e.,
\begin{equation}
    \label{ineq_L1}
    t_{s} \leq t \leq t_{e}
\end{equation}
\begin{equation}
    \label{ineq_C1}
    0 \leq MS_{i}(t-t_{s}) \leq MS_{i}(t_{e}-t_{s})
\end{equation}
 \begin{equation}
    \label{ineq_C2}
    0 \leq MS_{i}(t_{e}-t) \leq MS_{i}(t_{e}-t_{s})
\end{equation}
Subtracting Inequality \ref{ineq_C1} and \ref{ineq_C2}, we get inequality \ref{ineq_C3} which provides condition to define a lens $Lens_{i}$ as follows i.e., whether the radii intersection of $C_{1}$ $\cap$ $C_{2}$ $\geq$ 0.
\begin{equation}
    \label{ineq_C3}
    MS_{i}(t-t_{s}) - MS_{i}(t_{e}-t) \geq 0
\end{equation}

The bounded region in Figure \ref{fig:TimeSliceSingle1} shows $Lens_{i}$ i.e., intersection of $C_{1}$ $\cap$ $C_{2}$ $\geq$ 0. The bounded rendezvous regions at a given time instant $t$ is shown in Figure \ref{fig:TimeSliceDouble} with intersection of $Lens_{1}$ and $Lens_{2}$ $\subseteq$ $Ellipse_{1}$ $\cap$ $Ellipse_{2}$.

\begin{figure}[ht]
    \centering
    \includegraphics[width=0.30\textwidth]{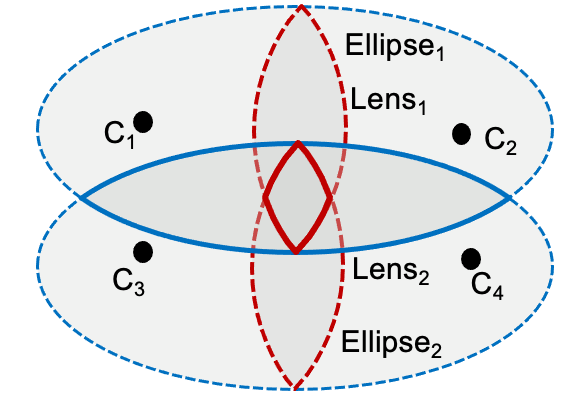}
    \caption{Time slicing model with lens intersection at time instant t}
    \label{fig:TimeSliceDouble}
\end{figure}

A rendezvous region using a time slicing model is defined as the intersection of $Lens_{i}$ and $Lens_{j}$ derived from $Ellipse_{1}$ $\in$ $G_{1}$ and $Ellipse_{2}$ $\in$ $G_{1}$. Figure \ref{fig:TimeSliceDouble} further represents the rendezvous region via $Lens_{1}$ $\cap$ $Lens_{2}$ at a given time instant $t$ which provides an even tighter bound as compared to the $Ellipse_{1} \cap Ellipse_{2}$. Lemma \ref{lemma:lemma1} and Lemma \ref{lemma:lemma2} provide formal proofs to the Theorem \ref{theorum:theorum1} which states that a time slice is a subset of the intersection of two space time prisms.

\begin{lemma}
\label{lemma:lemma1}
Given a pair <$G_{1}$,$G_{2}$> of gaps in trajectories in an isometric euclidean space, the space-time prism model bounds the area of possible rendezvous by the intersection of two ellipses $Ellipse_{1}$ and $Ellipse_{2}$ where $Ellipse_{1}$ specifies the possible locations during gap $G_{1}$ given maximum speed $MS_{1}$ and $Ellipse_{2}$ specifies the possible locations  during gap $G_{2}$ given maximum speed $MS_{2}$.
\label{lemma:Intersection_Ellipse}
\end{lemma}
\begin{proof} The proof is straightforward. Given a pair of intersecting trajectory gaps $G_{1}$ and $G_{2}$, any possible point $P$ $\in$ $\lbrace Ellipse_{1}, Ellipse_{2} \rbrace$ such that $Ellipse_{1}$ (P) $\geq$ 0 and $Ellipse_{2}$ (P) $\geq$ 0 must also satisfies inequality \ref{timerange_rendezvous}. Using the inequality \ref{timerange_rendezvous}, $P$ spatially qualifies within possible rendezvous location bounded by Equation \ref{ineq_E1intE2}:
 \begin{equation}
    \label{ineq_E1intE2}
    Ellipse_{1} \cap Ellipse_{2}(P) \geq 0 \subseteq \lbrace Ellipse_1 \ , \ Ellipse_2 \rbrace
\end{equation}

Equation \ref{ineq_E1intE2} shows a possible rendezvous area via intersection of $Ellipse_{1}$ and $Ellipse_{2}$ $\subseteq$ $\lbrace$$Ellipse_{1}$, $Ellipse_{2}$$\rbrace$ where $Ellipse_{1}$ (P) $\geq$ 0 and $Ellipse_{2}$ (P) $\geq$ 0. Hence, the intersection of space-time prisms bounds the area of possible rendezvous.
\end{proof}

\begin{lemma}
\label{lemma:lemma2}
Given a pair <$G_{1}$,$G_{2}$> of gaps in trajectories in isometric euclidean space, the time slicing model bounds the instantaneous area of possible rendezvous at time $t$ within the gap time interval by the intersection of $Lens_{1}$ and $Lens_{2}$ where $Lens_{1}$ specifies possible locations during gap $G_{1}$ given maximum speed $MS_{1}$ and Lens $Lens_{2}$ specifies possible locations during gap $G_{2}$ given maximum speed $MS_{2}$.
\end{lemma}
\begin{proof}
The circle intersection condition $C_{1}$ $\cap$ $C_{2}$ $\geq$ 0 is already satisfied from Inequality \ref{ineq_C3}. We can extend Inequality \ref{ineq_Lens1} and \ref{ineq_Lens2} for $Lens_{1}$ and $Lens_{2}$ respectively at a given time $t$:
\begin{equation}
    \label{ineq_Lens1}
    MS_{1}(t-t_{s}) - MS_{1}(t_{e}-t) \geq 0
\end{equation}
\begin{equation}
    \label{ineq_Lens2}
    MS_{2}(t-t_{s}) - MS_{2}(t_{e}-t) \geq 0
\end{equation}
Inequalities \ref{ineq_Lens1} and \ref{ineq_Lens1} qualifies any point $P$ such that $Lens_{1}$(P) $\geq$ 0 and $Lens_{2}$(P) $\geq$ 0. Using Lemma \ref{lemma:lemma1}, any point $P$ $\in$ $Lens_{1}$ $\cap$ $Lens_{2}$ $\subseteq$ $Ellipse_{1}$ $\cap$ $Ellipse_{2}$ $\subseteq$ $\lbrace$ $Ellipse_{1}$ , $Ellipse_{2}$ $\rbrace$ is bounded by instantaneous time $t$.
\end{proof}


\begin{theorem}
\label{theorum:theorum1}
Given a pair <$G_{1}$,$G_{2}$> of trajectory gaps in isometric euclidean space, the possible rendezvous area (i.e., lens intersection) bounded at any instant during the gap time interval is a subset of the possible rendezvous area (i.e., ellipse intersection) bounded by the space-time prism model.
\end{theorem}
\begin{proof}
Given points ($x_{s}$, $y_{s}$, $t_{s}$) and ($x_{e}$, $y_{e}$, $t_{e}$) are focii of the ellipse $E$, According to Equation \ref{ellipse}, a point (x,y) can lie anywhere in the geo-ellipse such that sum of distance from focii ($x_{s}$, $y_{s}$) and ($x_{e}$, $y_{e}$) $\leq$ ($t_{e}$-$t_{s}$). 
Equation \ref{ineq4} below derived from addition of Equation \ref{ineq2} and \ref{ineq3} of two circles also denotes a property of the ellipse.
\begin{equation}
    \label{ineq4}
    (x_{i} - x_{s})^{2} + (y_{i} - y_{s})^{2} + (x_{i} - x_{e})^{2} + (y_{i} - y_{e})^{2} \leq 2(t_{s} - t_{e})^{2}\times MS_{i}^{2}
\end{equation}
The left hand side of Inequality \ref{ineq4} also denotes each $Lens_{i}$ is valid $\forall$ t $\in$ $t_s \leq t_{i} \leq t_{e}$:
\begin{equation}
    \label{ineq5}
    \bigcup\limits_{t=t_{s}}^{t_{e}}Lens_{t} \subseteq \bigcup\limits_{t=t_{s}}^{t_{e}}(t_{s} - t_{e})^{2}\times MS_{i}^{2}
\end{equation}

In addition, using Lemma \ref{lemma:lemma1}, and \ref{lemma:lemma2}, the areal bounds defined by $Lens_{i}$ will not exceed the bounds defined by the geo-ellipse. 
\end{proof}


These lemmas and theorems can easily be generalized to spatial networks (e.g., road networks) by generalizing the ellipses and lenses to the subgraphs reachable during a gap time interval and the subgraphs reachable at a particular time instant within a gap time interval respectively, given a gap-start node, gap-end node, maximum speed and a gap time-interval. Due to lack of space, we are omitting the detailed proofs.


\section{Proposed Approach}
\label{sec:Proposed Approach}
In this section, we first explain some underlying concepts which are later used to define the output of the problem (i.e., \textit{possible rendezvous nodes}) in spatial networks. We then discuss the required pre-processing steps related to gathering candidate trajectory gap pairs which are later temporally sampled to construct spatial sub-networks. Finally, we present our proposed TGARD and DC-TGARD algorithms which leverage time slicing properties for better solution quality and computational efficiency. 

To calculate possible rendezvous locations, we first estimate the \textit{availability interval} for each object and then calculate how early that object is able to reach (i.e., \textit{earliest arrival time}) and the latest time the object can depart (\textit{latest departure time}) at a given node. We formally define them as follows:

An \textbf{Earliest Arrival Time} ($N^{EA}_{i}$) is the minimum time object $O_{i}$ takes from start node $N_{s}$ to an intermediate node $N_{u}$ i.e., $N_{u}^{EA}$ = $N_{u}^{EA}$ + $E_{w}$. A \textbf{Latest Departure Time} ($N^{EA}_{i}$) is the minimum time the object $O_{i}$ takes from end node $N_{d}$ to the intermediate node $N_{u}$ where $N_{u}^{LD}$ = $N_{u}^{LD}$ - $E_{w}$. An \textbf{availability interval ($\alpha$)} is the time interval that object $O_{i}$ wait at node $N_{u}$ such that $\alpha$(u) = [$N_{u}^{EA}$,$N_{u}^{LD}$] $\neq$ $\emptyset$, where $EA_{i}$(u) $\leq$ $LD_{i}$(u). A node $N_{u}$ is defined as \textbf{reachable} when $\alpha$($u$) $\neq$ $\emptyset$ and \textbf{not reachable} if $\alpha$(u) = $\emptyset$ or $N^{u}_{EA}$ $\leq$ $N^{u}_{LD}$. 

\begin{figure}[ht]
    \centering
    \includegraphics[width=0.30\textwidth]{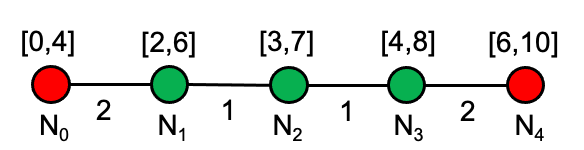}
    \caption{An illustration of Possible Rendezvous Nodes}
    \label{fig:TAG}
\end{figure}

Figure \ref{fig:TAG} shows nodes $N_{0}$ to $N_{4}$ with intervals [0,4], [2,6], [3,7], [4,8] and [6,10] respectively. Nodes $N_{0}$ and $N_{4}$ are start and end nodes respectively and nodes [$N_{1}$,$N_{2}$,$N_{3}$] $\in$ $N_{u}$. The earliest arrival is calculated by adding $E_{w}$ to a given intermediate node $N_{u}$. For instance, $N_{1}^{EA}$ = $N_{0}^{EA}$ + $E_{w}$($N_{0}, N_{1}$) = 0 + 2 = 2. Similarly, the values of $N_{2}^{EA}$, $N_{3}^{EA}$ and $N_{4}^{EA}$ are 3, 4 and 6 respectively. In contrast, the latest departure is calculated by subtracting $E_{w}$ to a given intermediate node $N_{u}$. For instance, ${N_{4}^{LD}}$ = $N_{5}^{LD}$ - d($N_{5}, N_{4}$) = 10 - 2 = 8. Similarly, the values of $N_{2}^{LD}$, $N_{3}^{LD}$ and $N_{4}^{LD}$ are 3, 4 and 6 respectively. 

In the case of multiple paths $Ps$ from a single source (i.e., $N_{s}$ or $N_{e}$), shortest path algorithms (e.g., \textbf{Bi-directional Dijkstra's}) or other breadth first search approaches are considered. Hence, multiple arrival times and late departure times need to be considered for a given node $N_{i}$ such that the earliest arrival time will be the \textbf{minimum} of all the arrival times derived via the shortest path from $N_{s}$ to $N_{u}$ (i.e., min([$N^{EA_{1}}_{u}$,$N^{EA_{2}}_{u}$,...,$N^{EA_{k}}_{u}$])). In contrast, the latest departure time will be the \textbf{maximum} of all the departure times that result via the shortest paths from $N_{e}$ to $N_{u}$ (i.e., max([$N^{LD_{1}}_{u}$,$N^{LD_{2}}_{u}$,...,$N^{LD_{k}}_{u}$])).
\begin{equation}
    \label{ineq:availability_interval}
    \alpha(N_{u}) = [min([{N}_{u}^{EA_{1}},..,N^{EA_{k}}_{u}]),max([{N}_{u}^{LD_{1}},..,N^{LD_{k}}_{u}])]
\end{equation}
A \textbf{Possible Rendezvous Node ($N_{R}$):} is defined as a possible node $N_{u}$ (or location) where two or more gaps belongs to different gaps could have physically met. For instance, the availability intervals of gaps $G_{i}$ and $G_{j}$ at $N_{u}$ are defined as $\alpha(N_{u},G_{i})$ and $\alpha(N_{u},G_{j})$ respectively. The rendezvous is feasible only if :
\begin{equation}
    \label{sufficient}
    \alpha(N_{u},G_{i}) \cap \alpha(N_{u},G_{j}) \geq TO
\end{equation}

where TO ($\neq$ $\emptyset$) is defined as \textbf{Time Overlap threshold}. 

If the above conditions are satisfied, the nodes are then the nodes are forwarded to human analysts for manual inspection. 



\subsection{Constructing Spatial Sub-networks}
Section \ref{sec:TimeSlicing} describes how time slicing in an unconstrained (i.e., euclidean) space narrows down the search space through the construction of a lens at given time instant $t$ to provide tighter bounds as compared to entire geo-ellipse region. 

In spatial networks, we restrict such unconstrained movements which closely resembles an vehicle's mobility in a road network. However, we use the combination of geo-ellipses and lenses to geometrically restrict the object movement capabilities in a Manhattan search space, thereby providing an upper bound towards the mobility. In addition to tightening the spatial bounds, time slicing also capture the time-dependent properties of spatial networks which more accurately model an object's mobility in real-world scenarios.

Algorithm \ref{Preprocess} first gathers trajectory gap pairs based on their temporal properties via a \textit{Filter} and \textit{Refine} approach to minimize the combinatorial computations. Then we temporally sample each trajectory gaps which captures a time-dependent edge which can later affects overall object's mobility.

\subsubsection{Extract Trajectory Gap Pairs:} 
\label{extracting_gaps}
For the effective gap collection, we use a \textit{plane-sweep approach} to sort the trajectory gaps $G_{i}$ and then filter out gaps that are not involved in a rendezvous. First, we sort all gaps $G_{i}$ $\in$ [$t_{s}$,$t_{e}$] based on the start time $t_{s}$ of the coordinates of $G_{i}$. Then, we filter gap pairs by checking the \textit{necessary} condition i.e., whether their respective time ranges overlap. For instance, Gap $G_{i}$ and $G_{j}$ should satisfy their respective time ranges [$t_{s}^{G_{i}}$,$t_{e}^{G_{i}}$] and [$t_{s}^{G_{j}}$,$t_{e}^{G_{j}}$] $\neq$ $\emptyset$ (using Equation \ref{timerange_rendezvous}). 

For the sufficient condition, we check if $G_{i}$ and $G_{j}$ are spatially intersecting such that $Ellipse_{i}$ $\cap$ $Ellipse_{j}$ $\neq$ $\emptyset$. If yes, then we save the resultant shape of $Ellipse_{i}$ $\cap$ $Ellipse_{j}$ and time range [$t_{s}^{G_{i}}$,$t_{e}^{G_{i}}$] $\cap$ [$t_{s}^{G_{j}}$,$t_{e}^{G_{j}}$] = [$t_{s}^{R}$,$t_{e}^{R}$] for creating spatial sub-networks.

\subsubsection{Creating Spatial Sub-Networks:} 
\label{creating_subnetworks}
Here, we generate spatial sub-networks based on the qualifying nodes within the spatial (i.e., $Ellipse_{1} \cap Ellipse_{2}$) and temporal [$t_{s}^{R}$,$t_{e}^{R}$] constraints of the trajectory gap pairs <$G_{1}$,$G_{2}$> as described in Section \ref{extracting_gaps}. To create a spatial network ($SN$), we perform a linear scan of ($N$) nodes and edges ($E$) $\in$ $Ellipse_{1} \cap Ellipse_{2}$ and to calculate edge weight $E_{w}$ by considering historic location traces residing within the time range of [$t_{s}^{R}$,$t_{e}^{R}$] and calculate the time cost by taking average speed $\mu$(MS) and dividing it by total edge distance between $N_{i}$ and $N_{j}$.

During time slicing, we capture the dynamic edge weights within [$t_{s}^{R}$,$t_{e}^{R}$]. First perform uniform sampling for a given temporal range into $N$ samples where [$t_{i}^{R},t_{i+1}^{R}$] $\in$ [$t_{s}^{R}$,$t_{e}^{R}$]. We then compute edge weights $E^{i}_{w}$ with given time range [$t_{i}^{R},t_{i+1}^{R}$] and filter nodes which are qualified within the $Lens_{i}$ where $N_{i}$ $\subseteq$ $N$. Hence for each time frame we generate subnetworks $SN_{i}$ with nodes $N^{i}_{u}$ edge weights $E^{i}_{w}$ $\forall$ i $\in$ [$t_{i}^{R},t_{i+1}^{R}$]. 

\begin{algorithm}
\caption{Spatial Sub-Networks with Rendezvous Gap Pairs}
\label{Preprocess}
\footnotesize
\begin{flushleft}
\hspace{\algorithmicindent}\textbf{Input:}\\
\hspace*{\algorithmicindent} Historic Trajectory Data (HTD)\\
\hspace*{\algorithmicindent} A set of Trajectory Gaps \lbrack$G_{1}$,...,$G_{n}$\rbrack\\
\hspace*{\algorithmicindent} A Spatial Network N\\
\hspace*{\algorithmicindent} A Sampling Rate K\\
\hspace*{\algorithmicindent}\textbf{Output:}\\
\hspace*{\algorithmicindent} Spatial Sub-Networks List\\
\end{flushleft}
\footnotesize
\begin{algorithmic}[1]
    \Procedure{:}{}
    \State {Sub-Networks Map $\leftarrow \emptyset$}
    \ForEach{$G_{i}$ $\in$ Non-Observed List}:
        \ForEach{$G_{j}$ $\in$ Observed list if Observed list $\neq$ $\emptyset$}
            \If{[$t_{s}^{G_{i}}$,$t_{e}^{G_{i}}$] $\cap$ [$t_{s}^{G_{j}}$,$t_{e}^{G_{j}}$] and $G_{i} \cap G_{j}$ $\neq \emptyset$}:
                \State [$SN_{0},SN_{1},..,SN_{k}$] $\leftarrow$ Subnetworks ($G_{i} \cap G_{j}$, HTD, N, K)
                \State Sub-Networks Map [<$G_{i},G_{j}$>] $\leftarrow$ [$SN_{0},SN_{1},..,SN_{k}$]
            \EndIf
        \EndFor
    \EndFor
    \State \Return Sub-Networks Map
\EndProcedure
\end{algorithmic}
\end{algorithm}

Algorithm \ref{Preprocess} summarizes the process of gathering trajectory gap pairs and creation of spatial networks in steps as follows:

\textbf{Step 1:} For a given trajectory gap $G_{i}$ in a Non-Observed list, we first check if $G_{i}$ intersects with $G_{j}$ in the Observed List $\not\in$ $\emptyset$. In case Observed List $\in$ $\emptyset$, we add $G_{i}$ $\in$ Observed List and move to $G_{i+1}$ in the Non-Observed List. If $G_{j}$ present in the Observed List, we check if the Ellipse bounds of $G_{i} \cap G_{j}$ and [$t_{s}^{G_{i}}$,$t_{e}^{G_{i}}$] $\cap$ [$t_{s}^{G_{j}}$,$t_{e}^{G_{j}}$] $\neq$ $\emptyset$. If not, then we derive intersected temporal range [$t_{s}^{R},t_{e}^{R}$] $\in$ [$t_{s}^{G_{i}}$,$t_{e}^{G_{i}}$] $\cap$ [$t_{s}^{G_{j}}$,$t_{e}^{G_{j}}$] and the resultant geo-ellipse intersection boundary of ($G_{i} \cap G_{j}$) for unconstrained rendezvous study area which is later projected on a the given spatial network $N$ to extract sub-network $SN$ for gap pairs <$G_{i}$,$G_{j}$> where $SN$ $\subseteq$ $N$. 

\textbf{Step 2:} For a given trajectory gap pair <$G_{i}$,$G_{j}$> and sampling rate K, we sample time interval [$t_{s}^{R},t_{e}^{R}$] in K samples such that each sub-network $SN_{i}$,$SN_{i+1}$,...,$SN_{k}$ has it's corresponding time-stamp [$t_{i}^{R},t_{i+1}^{R}$,...,$t_{k-1}^{R},t_{k}^{R}$]. Using historic trajectory data HTD, we calculate edge weights $E_{w}$ $\forall$ $SN$ using Equation \ref{edge_weight} where each $SN_{i}$ $\in$ $\{$ $E_{w}$, $t_{k}^{R}$ $\}$ is then saved in Sub-Network Map which is later returned as the output.

\subsection{Time Slicing Gap-Aware Rendezvous Detection Algorithm (TGARD)} 
The proposed TGARD algorithm captures an object's more realistic movements by considering parameters such as traffic congestion, which later affect the shortest path computation needed to calculate Earliest Arrival and Latest Departure time of the given sampled sub-network $SN_{k}$ where each sub-network $SN_{k}$ $\in$ [$t_{k}^{R}$] and $Lens_{k}$. Intermediate nodes $N_{u}$ is defined within $Lens_{k}^{G_{i}}$ $\cap$ $Lens_{k}^{G_{j}}$ towards which shortest path computation form start and end nodes is computed to calculate $N_{u}^{EA}$ and $N_{u}^{LD}$. However, we need to perform shortest path computation for every $N_{u}$ of every sampled sub-network $SN_{k}$ and edge weights $E_{w}^{k}$. Figure \ref{fig:TimeDependent} is a time aggregated graph representation \cite{gunturi2011critical} where each edge represents time series edge weights [$E_{w}^{t_{1}}$,$E_{w}^{t_{2}}$], and $E_{w}^{t_{1}}$ $\in$ $P_{1}$ and $E_{w}^{t_{2}}$ $\in$ $P_{2}$, which gives the shortest path from $N_{0}$ to $N_{6}$.

\begin{figure}[ht]
    \centering
    \includegraphics[width=0.35\textwidth]{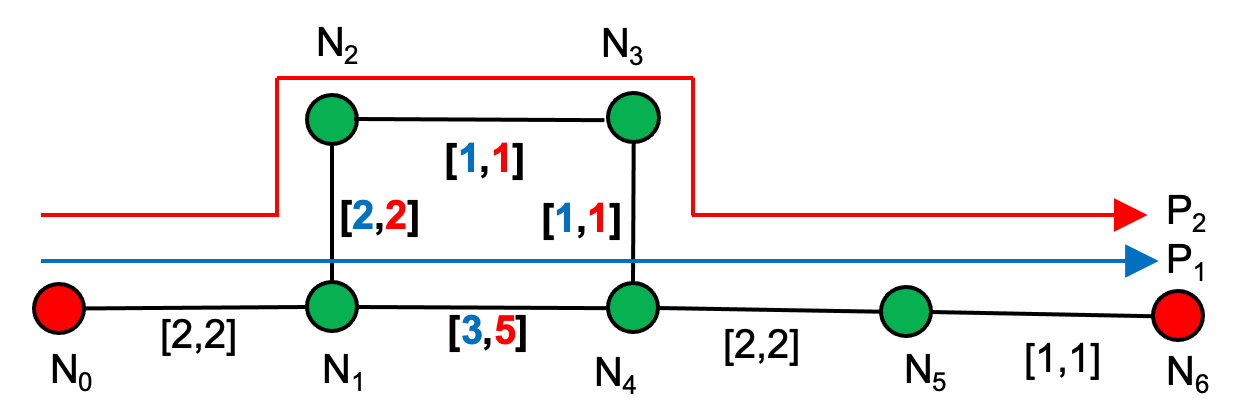}
    \caption{An illustration of Time Dependent Shortest Path}
    \label{fig:TimeDependent}
\end{figure}

In Figure \ref{fig:TimeDependent}, $P_{1}$ takes the path $N_{0}$, $N_{1}$, $N_{4}$, $N_{5}$, and $N_{6}$ at $t_{0}$ (since 2+1+1 $\geq$ 3). In contrast, path $P_{2}$ uses $N_{0}$, $N_{1}$, $N_{2}$, $N_{3}$, $N_{4}$, $N_{5}$ and $N_{6}$ since the path from $N_{1}$ and $N_{2}$ at $t_{1}$ is 5 $\geq$ 2+1+1 (i.e., $N_{1}$, $N_{2}$, $N_{3}$ and $N_{4}$). In this case, node $N_{4}$ gets affected and the shortest path needs to be recomputed towards $N_{4}$ at time $t_{2}$. To capture such changes in edge weights in  $E_{w}^{i}$ $\in$ $t_{i}$ and $E_{w}^{i+1}$ $\in$ $t_{i+1}$, we use a difference metric $\delta$($E_{w}^{i,i+1}$) and compare it with a threshold $\tau$ (i.e., whether $\delta$($E_{w}^{i,i+1}$) $\geq$ $\tau$. If it is, then we recompute the shortest path towards that node. In Algorithm \ref{TGARD}, we denote $\delta$($E_{w}^{i,i+1}$) as $\delta$($E_{w}^{k}$). For instance, in Figure \ref{fig:TimeDependent}, node $N_{4}$ gets affected at time $t=2$ given $\tau$ = 2 since $\delta$($E_{w}^{i,i+1}$) $\geq$ $\tau$. Hence we recompute the shortest path at $P2$ from $N_{0}$ to $N_{6}$ via $N_{2}$ and $N_{3}$. If not, the current path remains the same for future computations and thereby retains it's optimal sub-structure property for shortest path, stated as follows: 


\begin{theorem}
\label{theorum:theorum_SP}
Given $G = (V,E)$ with edge weights $E_{w}$. Let $P$ = $[N_{1}, N_{2}, N_{3}, ..., N_{k}]$ be the shortest path from $N_{1}$ to $N_{k}$ such that 1 $\leq$ i $\leq$ j $\leq$ k. If $P_{ij}$ = $[N_{i}, N_{i+1}, ..., N_{j}]$ is the sub-path of $P$ from node $N_{i}$ to $N_{j}$. Then, $P_{ij}$ is the shortest path from $P_{i}$ and $P_{j}$.
\end{theorem}
\begin{proof}
The proof is straightforwards. If we decompose $P$ into $P_{1,j}$ $\rightarrow$ $P_{i,j}$ $\rightarrow$ $P_{j,k}$, where $E_{w}$($P$) = $E_{w}$($P_{1,j}$) + $E_{w}$($P_{i,j}$) + $E_{w}$($P_{j,k}$) and assume a new path $P^{'}_{ij}$ from $N_{i}$ to $N_{j}$ such that $E_{w}$($P^{'}_{ij}$) $\leq$ $E_{w}$($P_{ij}$). Then, $E_{w}$($P$) $\leq$ $E_{w}$($P_{1,i}$) + $E_{w}$($P^{'}_{i,j}$) + $E_{w}$($P_{j,k}$) which contradicts $P$ as shortest path assumption from $N_{1}$ to $N_{k}$.
\end{proof}

\begin{algorithm}
\caption{Time Slicing Gap-Aware Rendezvous Detection}
\label{TGARD}
\footnotesize
\begin{flushleft}
\hspace{\algorithmicindent}\textbf{Input:}\\
\hspace*{\algorithmicindent} Time Overlap Threshold TO\\
\hspace*{\algorithmicindent} Rest same as Algorithm 1\\
\hspace*{\algorithmicindent}\textbf{Output:}\\
\hspace*{\algorithmicindent} Possible Rendezvous Nodes Hmap $N_{R}$
\end{flushleft}
\footnotesize
\begin{algorithmic}[1]
    \Procedure{:}{}
    \State Rendezvous Node Map RNMap $\leftarrow$ $\emptyset$
    \State {Use Algorithm 1 to construct Subnetwork Map SN Map}
    \ForEach{Sub Network $SN_{k}$ $\in$ Trajectory Gap Pairs $G_{i},G_{j}$}
        \State Construct $Lens_{k}^{G_{i}}$ and $Lens_{k}^{G_{j}}$ for $t_{k}^{R}$ $\in$ [$t_{s}^{R}$,$t_{e}^{R}$]
        \If{$Lens_{k}^{G_{i}}$ $\cap$ $Lens_{k}^{G_{j}}$ $\neq$ $\emptyset$ and $\delta$($E_{w}^{k}$) $\leq$ $\tau$}:
            \State $N_{u}$ $\leftarrow$ $Lens_{k}^{G_{i}}$ $\cap$ $Lens_{k}^{G_{j}}$
            \If{$N_{u}^{k}$ is reachable for Sub Network $SN_{k}$ and $\not\in$ RNMap}:
                \State Compute $N_{u}^{EA}$ and $N_{u}^{LD}$ for both $\alpha$($N_{u}^{G_{i}}$) and $\alpha$($N_{u}^{G_{j}}$)
                \If{$\alpha$($N_{u}^{G_{i}}$) $\cap$ $\alpha$($N_{u}^{G_{j}}$) $\neq$ $\emptyset$ and $\geq$ $TO$}:
                    \State Rendezvous Node Map RNMap[$N_{R}$] $\leftarrow$ $N_{u}$ $\forall$ $t_{k}^{R}$
                \EndIf
            \EndIf
        \EndIf
    \EndFor
    \Return Rendezvous Node Map RNMap
\EndProcedure
\end{algorithmic}
\end{algorithm}

The TGARD algorithm steps are as follows:

\textbf{Step 1:} First we compute a sub-network (SN) map for each gap pair <$G_{i}$,$G_{j}$> using Algorithm 1 and initialize a Rendezvous Node Map RNMap $\leftarrow$ $\emptyset$. For each gap pair <$G_{i}$,$G_{j}$>, we have $k$ sub-networks $SN_{k}$ which generate $Lens_{k}$ from both $G_{i}$ and $G_{j}$ at $t_{k}^{R}$. We then filter out $N_{u}^{k}$ residing within the spatial boundary of $Lens_{k}^{G_{i}} \cap Lens_{k}^{G_{j}}$, where $Lens_{k}^{G_{i}} \cap Lens_{k}^{G_{j}}$ $\neq$ $\emptyset$. In addition, we also check if edge weights $\delta$($E_{w}^{k}$) for sub-network $SN_{k}$ have changed considerably by comparing them with threshold $\tau$. If yes, then we consider all the affected nodes within $Lens_{k}^{G_{i}} \cap Lens_{k}^{G_{j}}$ for computing \textit{early arrival} and \textit{late departure} times to preserve completeness. If not, we re-use the previous shortest path calculations for \textit{early arrival} and \textit{late departure} at $t_{k-1}^{R}$ adhering to the optimal substructure property for a shortest path stated in Theorem \ref{theorum:theorum_SP}.

\textbf{Step 2:} After gathering all $N_{u}$ $\in$ $Lens_{k}^{G_{i}} \cap Lens_{k}^{G_{j}}$, we check if $N_{u}$ is reachable (i.e., $\alpha$($N_{u}$) $\neq$ $\emptyset$) by checking the availability interval [$N_{u}^{EA}$,$N_{u}^{LD}$] $\not\geq$ [$t_{s}$,$t_{e}$] for each gap $G_{i}$ and $G_{j}$. If valid, we then check whether both the availability intervals of $G_{i}$ and $G_{j}$ and their intersection is $\neq$ $\emptyset$ and $\geq$ $TO$. If yes, then $N_{u}^{k}$ qualifies as $N_{R}$ and is saved in a rendezvous node (RN) Map which is later returned as output.


\subsection{A Dual Convergence Approach (DC-TGARD)}
Since time slicing operation is computationally expensive, the Dual Convergence method leverage symmetric property of the ellipse to efficiently reduce the iterations of time slicing used by $TGARD$ while preserving correctness and completeness. We first define an early termination condition based on areal coverage of $Lens_{k}$, where $k$ $\in$ $0 \leq k \leq n$ for a given time frame $t$, where t $\in$ $0 \leq k \leq n$. We then perform \textit{ bi-directional pruning} using ellipse symmetric property and compute $EA$ and $LD$ from both tail-end of the ellipse in parallel to improve computational efficiency.

\textbf{Ellipse Symmetry Property:} Given ellipse with major and minor axis and center (0,0), then it's foci ($\pm$c,0) are equidistant from it's origin at c. Hence, the areal coverage drawn from the lenses centered at ($\pm$c,0) will be the same, permitting an equal number of nodes within the spatial bounds of the lenses. For instance, Figure \ref{fig:Dual Convergence} (b) shows an equal number of nodes residing within the spatial bounds of $Lens_{k}$ and $Lens_{n-k}$ and using Lemma \ref{theorum:theorum1}, time slicing does not leave out any other interpolated nodes (green) bounded within ellipse preserving the completeness of DC-TGARD. 
\begin{figure}[ht]
    \centering
    \includegraphics[width=0.45\textwidth]{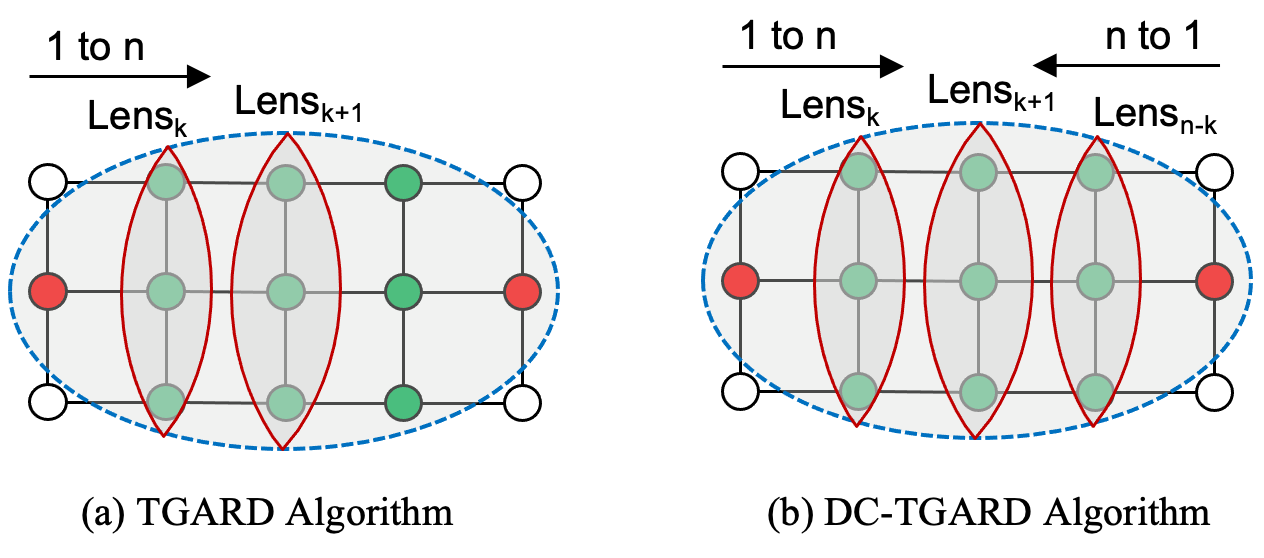}
    \caption{An illustration of TGARD vs Dual Convergence Method}
    \label{fig:Dual Convergence}
\end{figure}

\textbf{Early Stopping Criteria:} The early stopping criteria holds when the property of monotonicity (i.e., non-increasing) is violated while estimating the areal coverage of lenses throughout every time slicing operation from $t_{s}$ to $t_{e}$. For instance, Figure \ref{fig:Dual Convergence} (a), shows the increasing areal-coverage i.e., A($Lens_{k}$) $\leq$ A($Lens_{k+1}$) (given A($Lens_{k}$) as the areal coverage of lens $k$) at $t_{k}$ and $t_{k+1}$ respectively (where, $t_{k+1}$ $\geq$ $t_{k}$). This is due to the increasing (or non-decreasing) length of the minor axis which preserves the monotonicity until the $Lens_{k+1}$ of the ellipse. After $Lens_{k+1}$, the areal coverage start decreasing and thereby violating the monotonic property i.e., (A($Lens_{k}$)$\leq$A($Lens_{k+1}$)$\geq$A($Lens_{i+2}$)) and, resulting in early stopping of the time slicing operation.

\textbf{Dual Convergence Approach:} We need to consider each time-slice operation to preserve completeness of the algorithm. The dual convergence approach introduces \textbf{bi-directional pruning} in conjunction to efficiently converge towards \textit{early stopping criteria}. Hence, a time-slicing operation is performed from both tail ends of the ellipse focii in parallel such that both lenses generated at $t$ and $t_{e}-t$ follows the property of \textbf{ellipse symmetry}. While performing bi-directional pruning, we check the \textit{early stopping criteria} if A($Lens_{i}$) $\geq$ A($Lens_{i+1}$) such that the monotonic property is violated. Figure \ref{fig:Dual Convergence} shows an example of dual convergence where baseline TGARD in Figure \ref{fig:Dual Convergence} (a) performs a linear time slicing operation from $Lens_{k}$, $Lens_{k+1}$ to $Lens_{n}$ covering all the interpolated nodes. Figure \ref{fig:Dual Convergence} (b) shows bi-directional pruning at $Lens_{k}$ and $Lens_{n-k}$ where k $\in$ [1,n]. The areal coverage of $Lens_{k}$ i.e., A($Lens_{k}$) and A($Lens_{n-k}$) using the property of ellipse symmetry computed in parallel and further checking if A($Lens_{k}$) $\geq$ A($Lens_{k+1}$) which gives an early stopping condition. For instance, in Figure \ref{fig:Dual Convergence} DC-TGARD terminates when A($Lens_{k+1}$) $\geq$ A($Lens_{n-k}$) at $t$=2 as compared to TGARD in Figure \ref{fig:Dual Convergence} (a) which terminates at $t$=3. Hence, the areal coverage pruning is more efficient in DC-TGARD as compared to TGARD. Hence the pruning is similar to to a \textit{bitonic array} where the peak areal coverage using $Max Overlap$ variable. 

\begin{algorithm}
\caption{Dual Convergence Time Slicing Gap-Aware Rendezvous Detection (DC-TGARD)}
\label{DCTGARD}
\footnotesize
\begin{flushleft}
\hspace{\algorithmicindent}\textbf{Input:}\\
\hspace*{\algorithmicindent} Same as Algorithm 2\\
\hspace*{\algorithmicindent}\textbf{Output:}\\
\hspace*{\algorithmicindent} Possible Rendezvous Nodes List $N_{R}$
\end{flushleft}
\footnotesize
\begin{algorithmic}[1]
    \Procedure{:}{}
    \State {Use Algorithm 1 to construct Subnetwork Map SN Map}
    \State Rendezvous Node Map RNMap $\leftarrow$ $\emptyset$ and Max Overlap $\leftarrow$ $\emptyset$
    \State Construct Sub Network $SN_{k}$ and $SN_{n-k}$ for $t^{R}_{k}$ and $t^{R}_{n-k}$ respectively 
    \ForEach{Sub Network $SN_{k}$ and $SN_{n-k}$ $\in$ Trajectory Gap Pairs <$G_{i},G_{j}$>}
    \State Construct $Lens_{k}^{G_{i}}$, $Lens_{k}^{G_{j}}$ $\in$ $t_{k}^{R}$ and $Lens_{n-k}^{G_{i}}$, $Lens_{n-k}^{G_{j}}$ $\in$ $t_{n-k}^{R}$
        \While{Max Overlap $\not\geq$ Area($Lens_{k}^{G_{i}}$ $\cap$ $Lens_{k}^{G_{j}}$)}
            \If{$Lens_{k}^{G_{i}}$ $\cap$ $Lens_{k}^{G_{j}}$ $\neq$ $\emptyset$ and $\delta$($E_{w}^{k}$, $E_{w}^{n-k}$) $\leq$ $\tau$}
                \State $N_{u}^{k}$ $\leftarrow$ $Lens_{k}^{G_{i}}$ $\cap$ $Lens_{k}^{G_{i}}$ and $N_{u}^{n-k}$ $\leftarrow$ $Lens_{n-k}^{G_{i}}$ $\cap$ $Lens_{n-k}^{G_{j}}$
                \If{$N_{u}^{k}$ $\in$ $SN_{k}$ and $N_{u}^{n-k}$ $\in$ $SN_{n-k}$ reachable and $\not\in$ RNMap}:
                    \State Compute $N_{u}^{EA}$, $N_{u}^{LD}$ $\forall$ $t_{k}^{R}$ and $t_{n-k}^{R}$ $\in$ $\alpha$($N_{u}^{G_{i}}$) and $\alpha$($N_{u}^{G_{j}}$)
                    \If{$\alpha$($N_{u}^{G_{i}}$) and $\alpha$($N_{u}^{G_{j}}$) $\neq$ $\emptyset$ and $\geq$ $TO$ $\forall$ $t_{k}^{R}$ and $t_{n-k}^{R}$}:
                        \State Rendezvous Node Map RNMap[$N_{R}$] $\leftarrow$ $N_{u}^{k}$ and $N_{u}^{n-k}$
                    \EndIf
                \EndIf
            \EndIf
        \EndWhile
    \EndFor
    \State \Return Rendezvous Node Map RNMap
\EndProcedure
\end{algorithmic}
\end{algorithm}





The DC-TGARD algorithm steps are as follows:



\textbf{Step 1:} After generating Sub-Network Map (SN Map) via Algorithm 1, we first initialize Rendezvous Node Map RNMap and Max Overlap variable as $\emptyset$. We then generate sub-networks $SN_{k}$ and $SN_{n-k}$ for each $t_{k}^{R}$ and $t_{n-k}^{R}$ and generate $Lens_{k}$ and $Lens_{n-k}$ $\subseteq$ $SN_{k}$ and $SN_{n-k}$ respectively. We then check if $Lens_{k}^{G_{i}}, Lens_{n-k}^{G_{i}} \cap Lens_{k}^{G_{j}}, Lens_{n-k}^{G_{j}}$ $\leftarrow$ $\emptyset$ and simultaneously check $\delta$($E_{w}^{k}$, $E_{w}^{n-k}$) with threshold $\tau$. Rest of the steps are similar to TGARD instead we simultaneously calculate every variables for early arrival, late departures and availability intervals along with their respective conditions.

\textbf{Step 2:} We update the Max Overlap variable with Area of $Lens^{G_{i}}_{k}$ $\cap$ $Lens^{G_{j}}_{k}$ or ($Lens^{G_{i}}_{n-k}$ $\cap$ $Lens^{G_{j}}_{n-k}$). Finally, we check the early termination condition (i.e., current $Lens^{G_{i}}_{k}$ $\cap$ $Lens^{G_{j}}_{k}$ and $Lens^{G_{i}}_{n-k}$ $\cap$ $Lens^{G_{j}}_{n-k}$ $\geq$ Max Overlap. If yes, then the loop has reaches to the \textit{peak element} and algorithm terminates.

\section{Theoretical Evaluation}
\label{sec:TheoreticalEvaluation}
\subsection{Correctness and Completeness} In this section, we provide a theoretical analysis of the correctness and completeness of the proposed algorithms.
\begin{lemma}
TGARD and DC-TGARD algorithms are correct.
\end{lemma}
\begin{proof}
Given a finite set of $|K|$ trajectories with maximum $N$ number of finite points resulting in a maximum bound of $|K|\times|N|$ finite trajectory points. Hence, finite number of trajectory gaps $G_{i}$ generated while performing pre-processing for finite trajectory gap-pair generation resulting a finite number of operations will be performed by \textit{TGARD and DC-TGARD} algorithms to terminate at a finite time. The correctness of the algorithm also depends upon the extent of the overlap of the two availability intervals $\alpha$(1) $\in$ $G_{1}$ and $\alpha$(2) $\in$ $G_{2}$. For instance, a low overlap threshold results in multiple false positive rendezvous whereas a high overlap threshold results in more false negatives. Hence, both \textit{TGARD and DC-TGARD} algorithms are correct for a given overlap threshold.
\end{proof}

\begin{lemma}
TGARD and DC-TGARD algorithms are complete.
\end{lemma}
\begin{proof}
The proposed \textit{TGARD and DC-TGARD} algorithms consider all the nodes within the spatial bounds of the lenes and their respective geo-ellipses. Theorem \ref{theorum:theorum1} proved that a time slice lens is a subset of geo-ellipse and does not exceed it's limit. This ensures that each node within gap is participating during time slicing operation. Hence, both \textit{TGARD and DC-TGARD} algorithms are complete.
\end{proof}

\subsection{Asymptotic Analysis}
The time complexity of the comparison operations for each algorithm are as follows: 

\textbf{Exacting Candidate Gap Pairs:} For generating candidate pairs, $\binom{|K|}{2}$ trajectories must be selected where each such pair will have $|N-1|\times|N-1|$ comparisons. Hence, the total number of comparisons will be $\binom{|K|}{2}\times|N-1|\times|N-1|$, which results in an asymptotic worst case of $O(|K|^{2}\times|N|^{2})$.

\textbf{Creating Sub-Networks (SNs):} For a given subgraph $SN$ with $|E|$ $\subseteq$ $N\times N$ edges, the edge weights needs to be calculated $|T|$ times where $E_{w}$ is calculated at constant time. Hence the overall time complexity for calculating $E_{w}$ is O ($|E|\times |T|$).

\textbf{TGARD Algorithm:} Given a sub-network $SN$ for $N$ nodes, we perform a linear scan for each time slice $|T|$ to gather intermediate node $N_{u}$ where $N_{u}$ $\leq$ $N$ resulting in $O(|K|\times|T|)$ operations. For computing Early Arrival and Late Departure for $N_{i}$, we run a \textit{bi-directional Dijkstra's} algorithm with $O(|N_{i}|+|E|)$. Hence, the algorithm perform $O(|K|\times|T|)$ + $O(|N_{u}|\times|E|)$ operations. 

\textbf{DC-TGARD Algorithm:} The worst case of DC-TGARD will be similar to $TGARD$ except we 
perform bi-directional searches while computing the time slice operations and the early termination condition guarantees the algorithm stops in $T/2$ steps. This in turn reduces shortest path computations, which enhances computational efficiency in high density regions.

\section{Experimental Evaluation}
\label{sec:ExperimentalEvaluation}

The goal of the experiments was to validate the benefit of the proposed time slicing approach for reducing the search space of rendezvous detection. We evaluated the solution quality by comparing the proposed DC-TGARD against space-time prism methods \cite{uddin2017assembly}. We also compared the execution time of TGARD and DC-TGARD under different parameters. Details shown in Fig \ref{fig:ExperimentDesign}.

\begin{figure}[ht]
    \centering
    \includegraphics[width=0.45\textwidth]{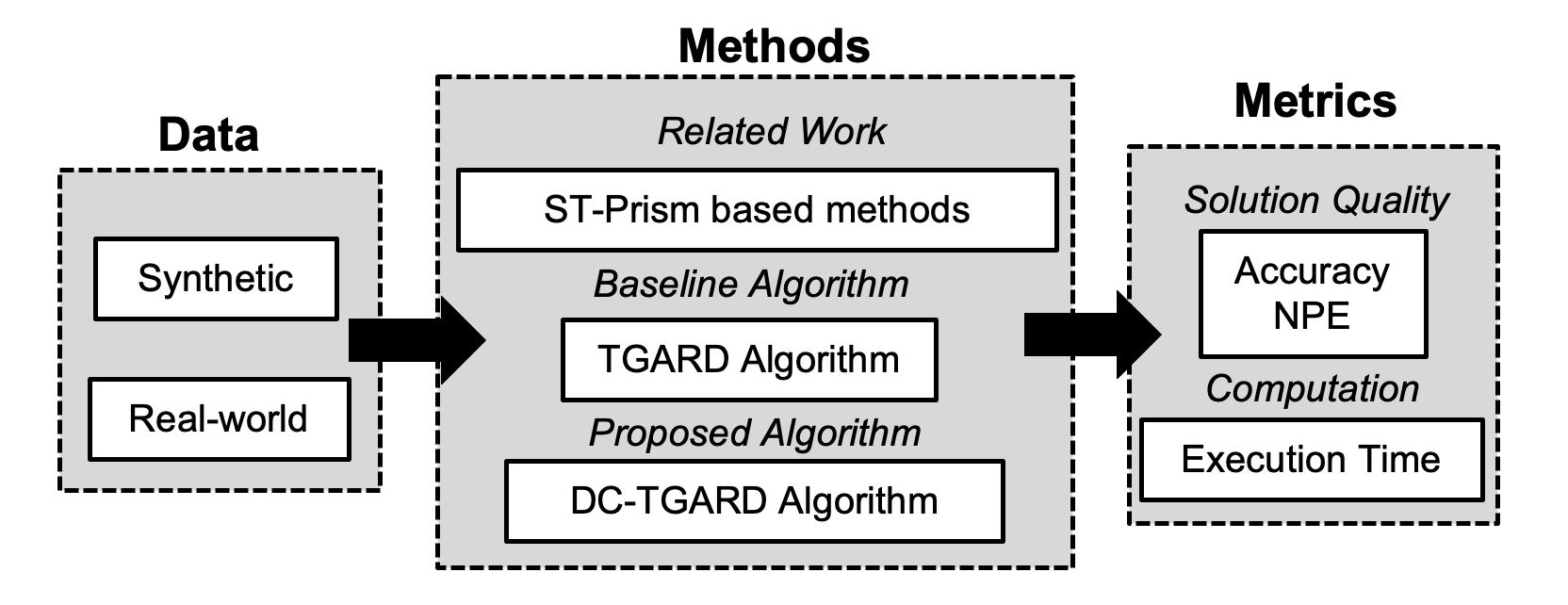}
    \caption{Experiment Design}
    \label{fig:ExperimentDesign}
\end{figure}

\textbf{Real World Data:} We used  Geolife \cite{zheng2010geolife} dataset based on Beijing Road Networks where each location has latitude, longitude, and height with a variety of travel modes (e.g., walking, driving, etc.). In this paper, we limited our evaluation to driving patterns due to their accordance with road network topology (e.g., road segment and intersection). In addition, we further simulated certain trajectories by adding more objects and case scenarios of rendezvous patterns to test the effectiveness and scalability of the proposed methods.

\textbf{Synthetic Data Generation}: For \textit{solution quality} experiment, we lacked ground truth data (i.e., information on whether a node was involved in a rendezvous or not). Therefore, we evaluated the proposed algorithms on synthetic data derived from the Geolife dataset. First, we gather trajectories on a fixed study area of the Beijing road network with mobility data. We then pre-processed the trajectory points with gap durations greater than 30 mins and randomly classified each gap based on whether $N_{r}$ was reachable or was not using the proposed methods. 

\textbf{Computing Resources:} We performed our experiments on a system with a 2.6 GHz 6-Core Intel Core i7 processor and 16 GB 2667 MHz DDR4 RAM.

\subsection{Solution Quality}
To assess solution quality, we considered bounding efficiency and accuracy. We developed an efficiency metric called node pruning efficiency (NPE), to measure the tightness of the proposed filter in the trajectory gap pairs. We defined NPE as the ratio of the nodes of a total study area to the nodes within a bounded region (Eq. \ref{ineq10}): 
\begin{equation}
    \label{ineq10}
    NPE = \frac{Nodes\ in\ Total\ Study\ Area}{Nodes\ in\ the\ Bounded\ Region}
\end{equation}

\setlength{\belowcaptionskip}{-10pt}
\begin{figure}
\begin{multicols}{2}
\label{fig:SolutionQualityNPE}
    \includegraphics[width=\linewidth]{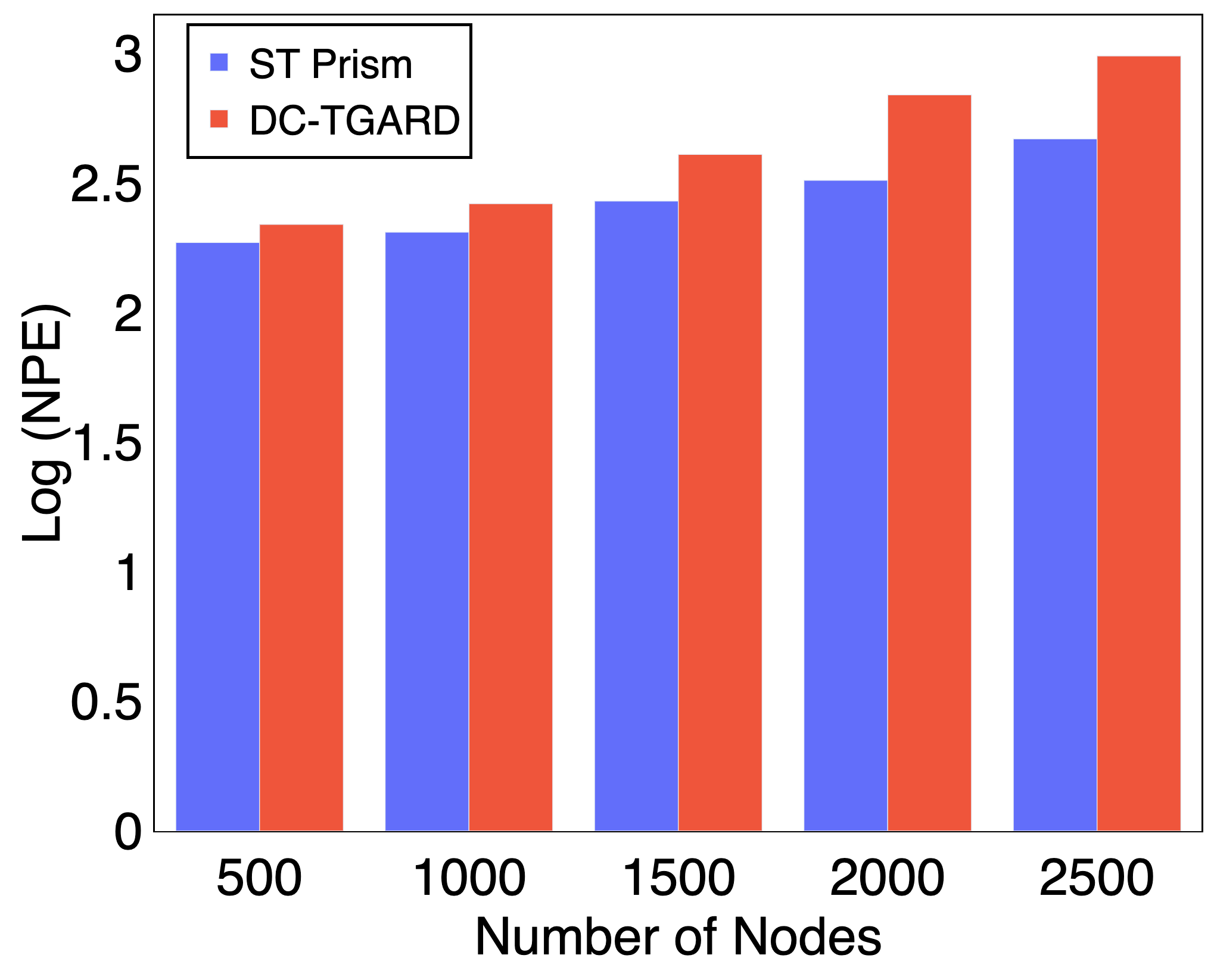}\par
    (a) Number of Nodes\\
    \includegraphics[width=\linewidth]{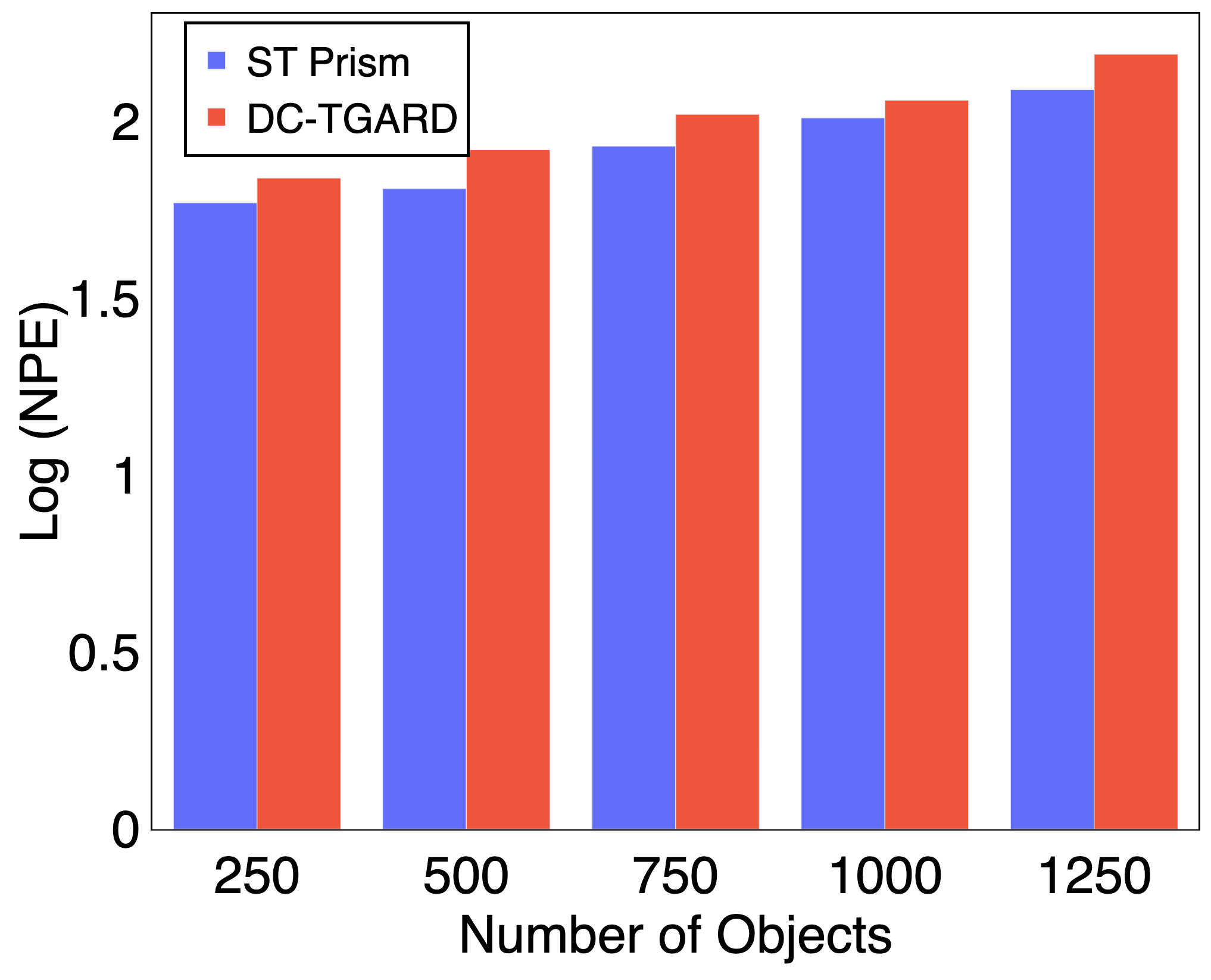}\par 
    (b) Number of Objects\\
    \end{multicols}
\begin{multicols}{2}
    \includegraphics[width=\linewidth]{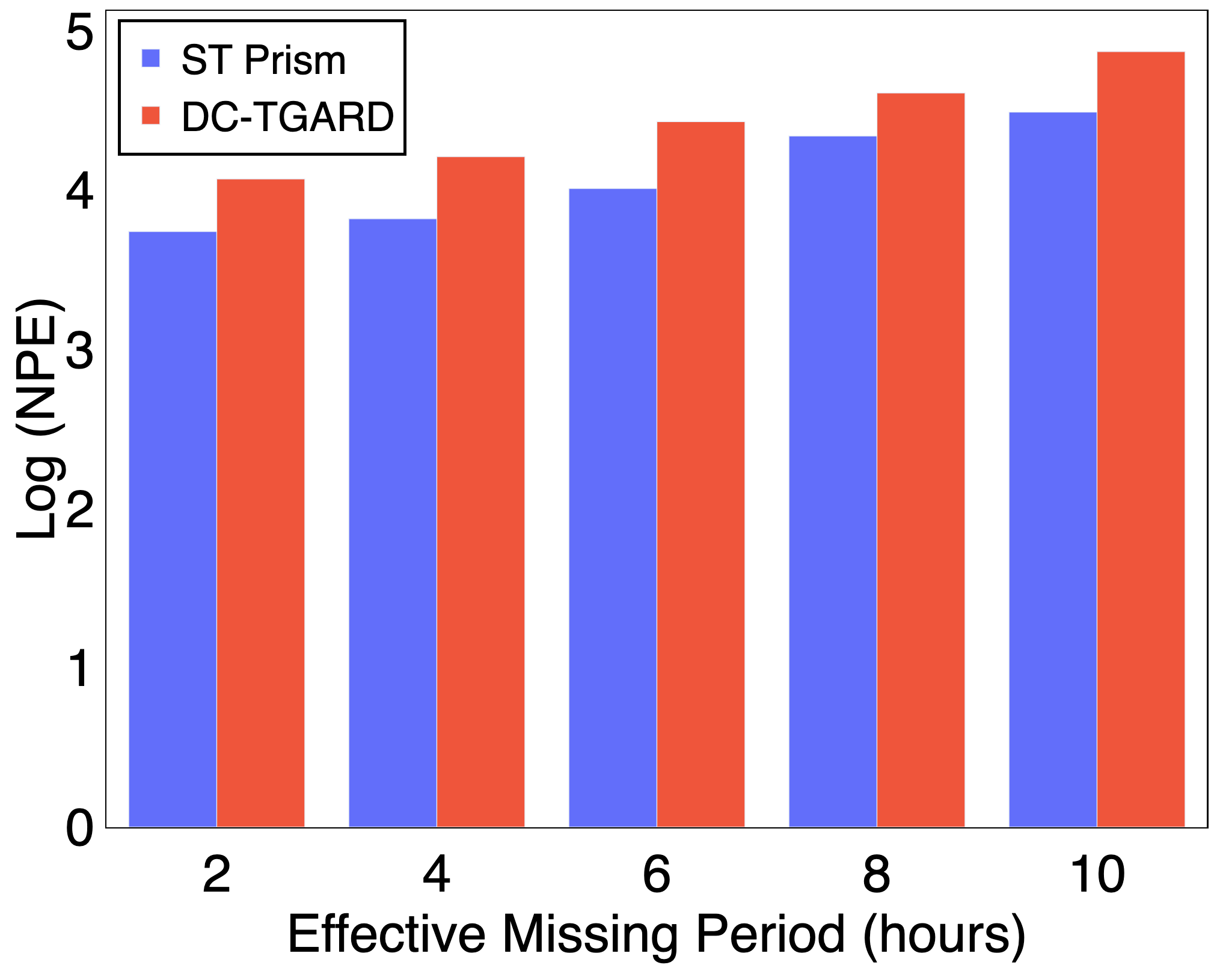}\par
    (c) Effective Missing Period\\
    \includegraphics[width=\linewidth]{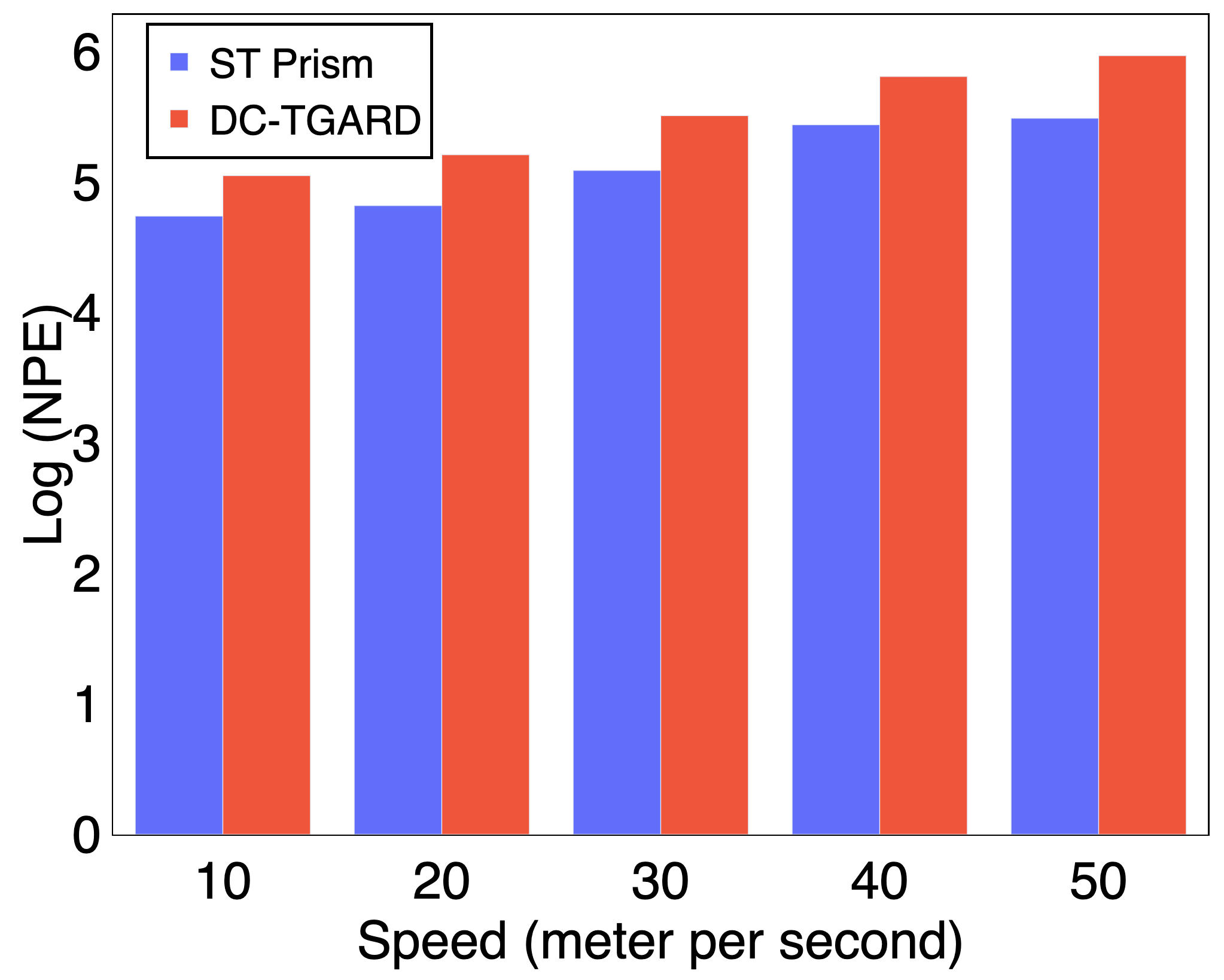}\par
    (d) Speed\\
\end{multicols}
\caption{DC-TGARD performs better than a space-time prism method under different parameters}
\end{figure}

\textbf{Nodes Pruning Efficiency (NPE):} We fixed the study area and varied the Nodes (N), Effective Missing Period (EMP), Speed (MS) and Number of Objects (O) within each trajectory gap $G_{i}$. Figure 9 shows the results for the space-time prism approach and the proposed time slicing approach. As can be seen in Figure 9 (a), DC-TGARD has better node pruning efficiency (NPE) as we increase the number of nodes (since node density increases). A similar trend is seen in Figure 9 (b) except the effectiveness of time slicing is not as significant as we increase the number of objects of a fixed study area. Figure 9 (c) and 9 (d) again show the time slicing effectiveness is superior as compared to baseline approach as we increase the duration of the effective missing period (EMP) of gap.


\textbf{Accuracy:} We used synthetic data with manually labelled ground truth data about whether or not the objects were involved in rendezvous for a fixed number of objects (i.e., 2000) within a fixed size network of 5000 nodes. We then varied different parameters including Time Overlap Threshold ($TO$). 

\textbf{(1) Number of Objects:} We set the EMP to $4$ hours, MS to 30 $m/s$, Time Overlap Threshold (T) to 30 mins and we varied the number of objects from $500$ to $1250$. Figure 10 (a) shows DC-TGARD gives a more accurate representation of the minimum travel time from the start node and end node of the trajectory gap compared to space-time prism based methods. This results in more accurate estimates of $N_{r}$ and reduces the number of false negatives, which proved to be quite common in the space-time prism based approaches \cite{uddin2017assembly}.

\textbf{(2) Effective Missing Period (EMP):} Next, we set the number of objects at $500$, the number of nodes of $5 x {10}^{4}$,  MS at 30 $m/s$, TO threshold to 30 mins, and varied the EMP threshold from $2$ to $10$ hours. As shown in Figure 10 (b), we find that DC-TGARD outperforms space time prisms. Higher EMP estimates increase the size of the geo-ellipses resulting in higher false positives. 

\textbf{(3) Speed:} We kept the number of objects and EMP constant at $1000$ and $4$ hours respectively, the TO threshold at $30$ mins and increased the maximum speed MS from $10$ to $50$ in $meter/second$ units. Figure 10 (c) shows DC-TGARD is more accurate as compared to ST-Prisms. The reason is that high speed outputs a greater number of potential rendezvous. This increases false positives on both algorithms but DC-TGARD filters out \textit{non-reachable} nodes, thereby increasing accuracy.

\textbf{(4) Time Overlap Threshold (TO):} We set  EMP at $4$ hours, MS at 30 $m/s$ and varied the $TO$ threshold from $20$ min to $100$ mins. Again, DC-TGARD outperforms space-time prism as we increase TO threshold (Figure 10 (d)). We also see that greater TO threshold results in less $N_{r}$, resulting in more false negatives which are more accurately captured by the proposed DC-TGARD algorithm. 
\setlength{\belowcaptionskip}{-10pt}
\begin{figure}
\begin{multicols}{2}
    \includegraphics[width=\linewidth]{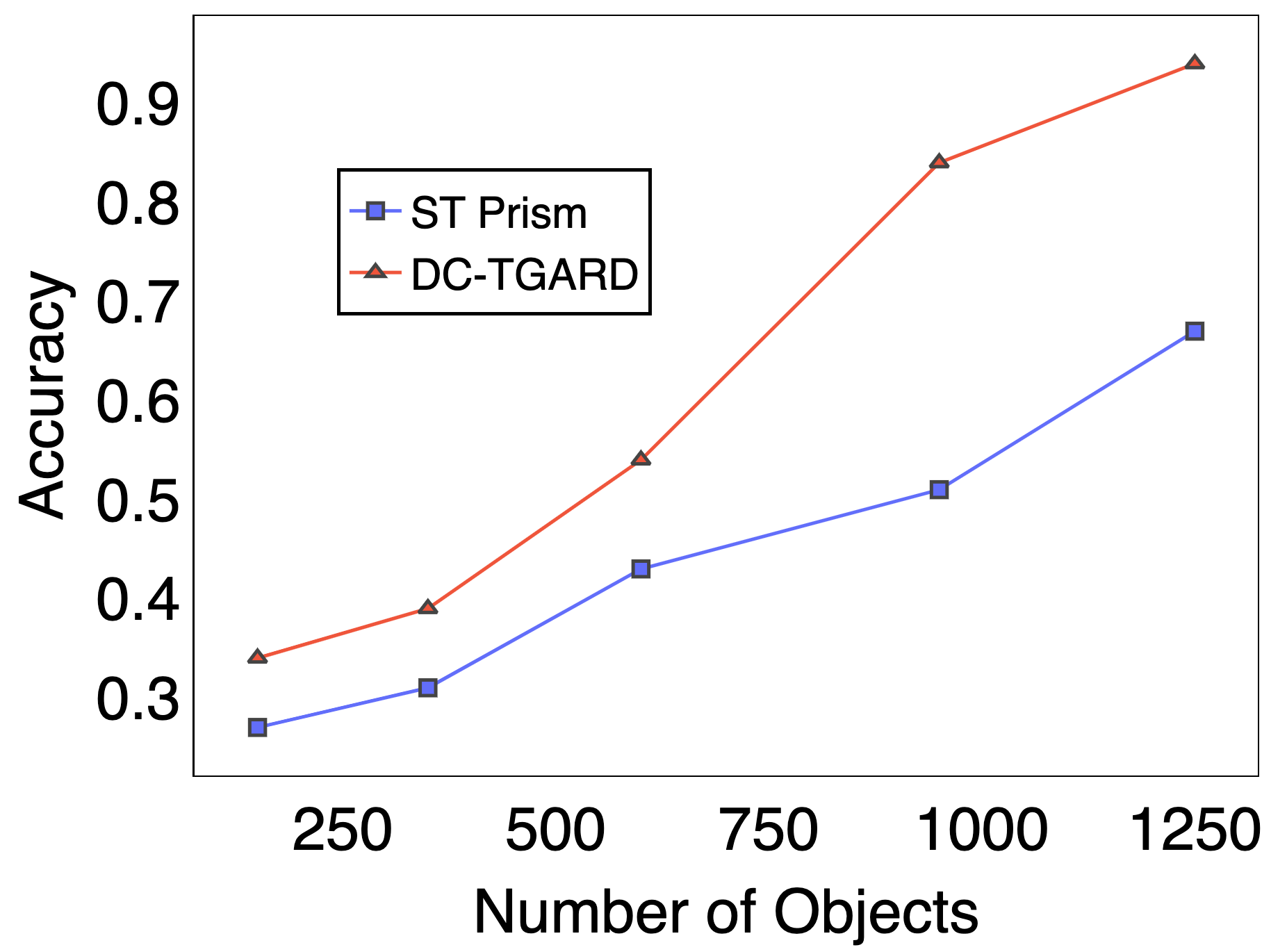}\par
    (a) Number of Objects\\
    \includegraphics[width=\linewidth]{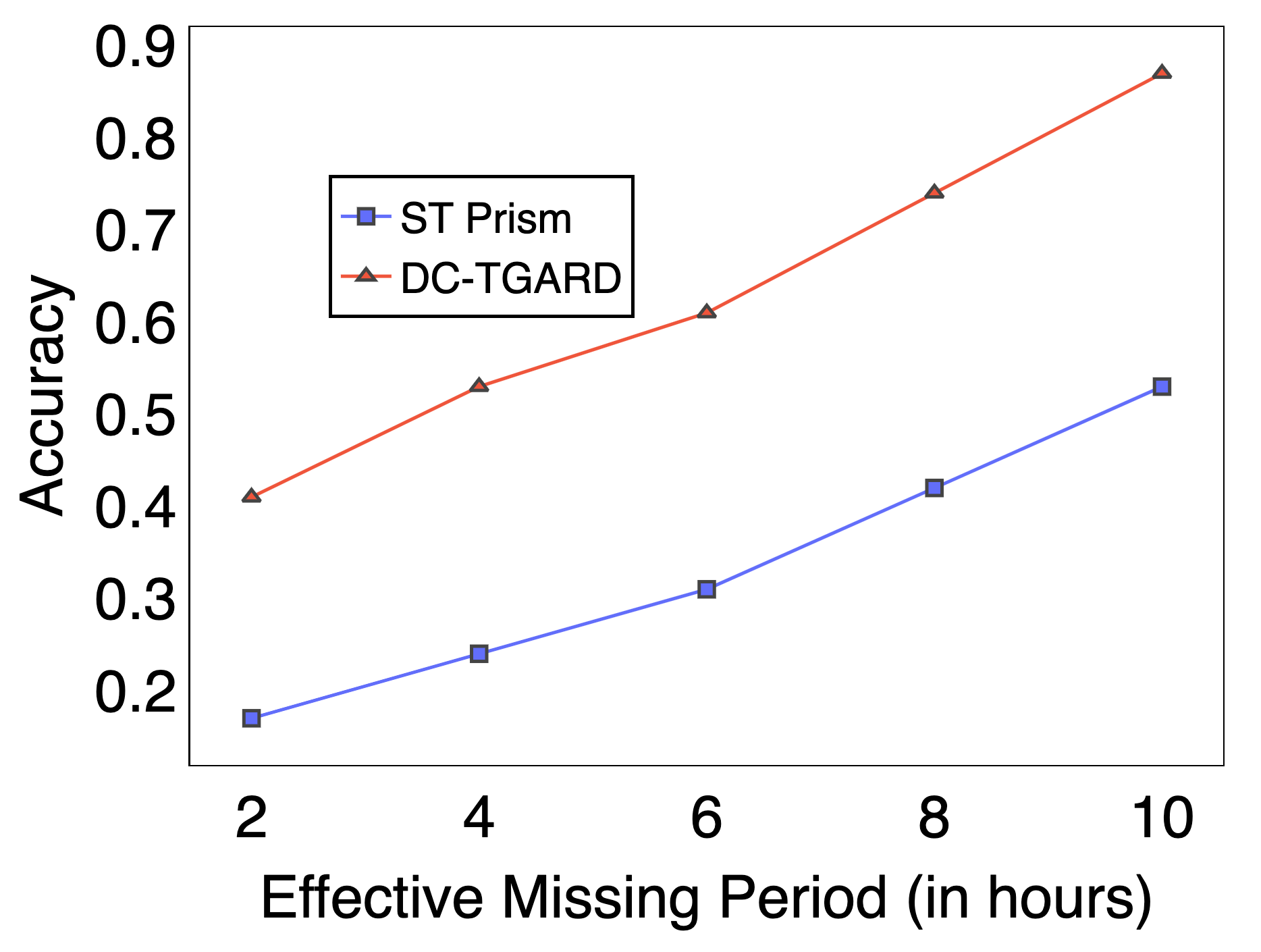}\par 
    (b) Effective Missing Period\\
    \end{multicols}
\begin{multicols}{2}
    \includegraphics[width=\linewidth]{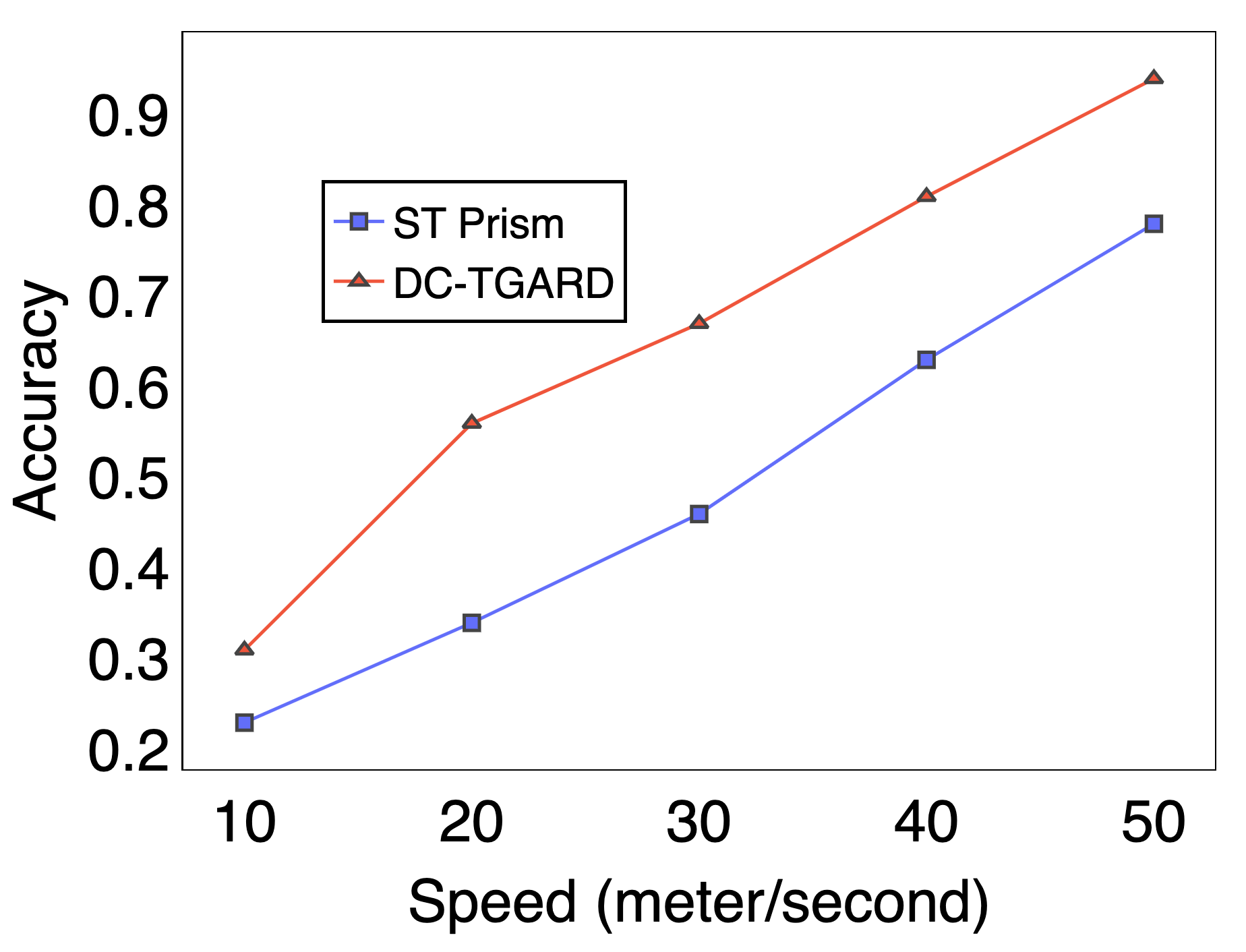}\par
    (c) Speed\\
    \includegraphics[width=\linewidth]{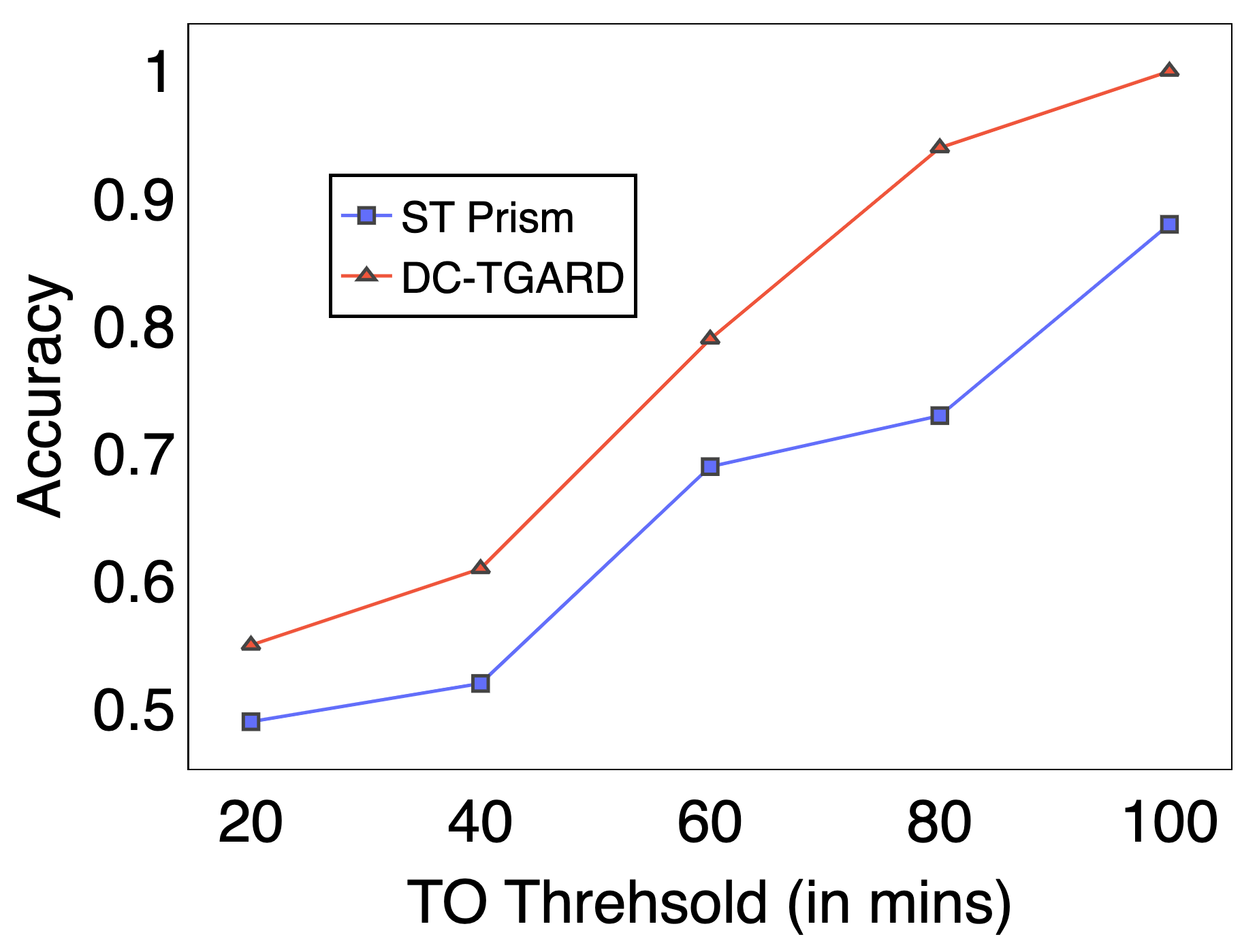}\par
    (d) TO Threshold (TO)\\
\end{multicols}
\label{fig:SolutionQuality}
\caption{DC-TGARD is more accurate ST Prism under varying parameters}
\end{figure}


\subsection{Computational Efficiency}
We then compared the proposed DC-TGARD against the baseline TGARD algorithm on computation efficiency using Number of Nodes, Number of Objects, Effective Missing Period, Speed MS and Time Overlap Threshold (TO).

\textbf{(1) Number of Objects:} We set the EMP to $4$ hours, MS to 30 $m/s$, Time Overlap Threshold (T) to $30$ mins and we varied the number of objects from $500$ to $2500$. Figure 11 (a) shows that DC-TGARD always runs faster than TGARD algorithm it's dual-convergence time-slicing operation leverages early stopping criteria. 

\textbf{(2) Effective Missing Period (EMP):} We set the number of objects at $500$, the number of nodes of $5 x {10}^{4}$,  MS at 30 $m/s$, TO threshold to 30 mins, and varied the EMP from $30$ to $90$ minutes. As shown in Figure 11 (b), we find that DC-TGARD outperforms TGARD with increasing EMP. The bi-directional pruning of the dual convergence approach reduces the time slicing operations.

\textbf{(3) Speed:} We kept the number of gaps and EMP constant at $1000$ and $4$ hours respectively, TO threshold 30 mins and increased the maximum speed MS from $10$ to $50$ in $m/s$ units. Figure 11 (c) shows DC-TGARD is faster than TGARD.  Hence, the dual convergence of DC-TGARD helps but not that significantly if we increase the MS. The reason is similar to $EMP$ results since higher maximum speed produces larger ellipses except non-reachable nodes the non-reachable nodes decreases for both the algorithms which results in more shortest path computations for both algorithms. However, DC-TGARD outperform due to it's dual convergence property, the execution times for both algorithms are linear as the speed increases. 

\textbf{(4) Time Overlap Threshold (TO):} We again kept the number of objects at 1000, the EMP at $4$ hours, MS at 30 $m/s$ but this time we increased the TO threshold from 20 min to 100 mins. Again, Figure 11 (d) shows DC-TGARD outperforms TGARD as we increase TO threshold i.e., the dual convergence nature of DC-TGARD efficiently filter out more $N_{r}$ and fewer shortest path computations are performed.

\textbf{(5) Number of Nodes:} We did a high density test by setting the number of objects to $1000$ and varied the number of nodes from $500$ to $2500$. For low density test, we set the number of objects to $100$ and varied number of nodes from $200$ to $1000$. For both experiments, EMP was set to $4$ hours, MS to 30 $m/s$ and Time Overlap Threshold (TO) to $40$ mins for a fixed study area. Figure 5 (e) and (f) shows that DC-TGARD always outperforms the TGARD algorithm. It's time slicing operations are effective for both high and low density road networks.


\begin{figure}
\label{fig:computationalefficiency}
\begin{multicols}{2}
    \includegraphics[width=\linewidth]{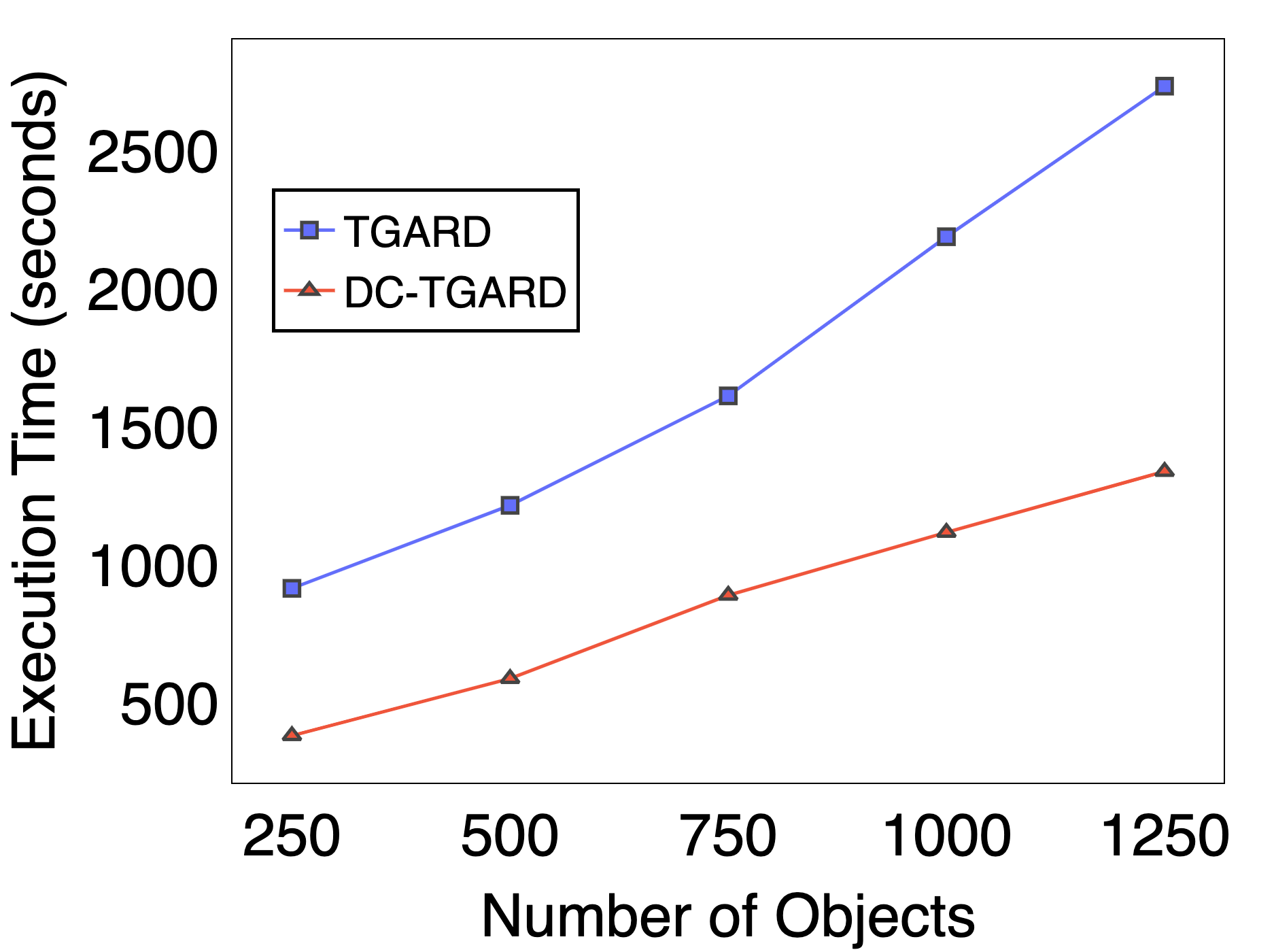}\par
    (a) Number of Objects\\
    \includegraphics[width=\linewidth]{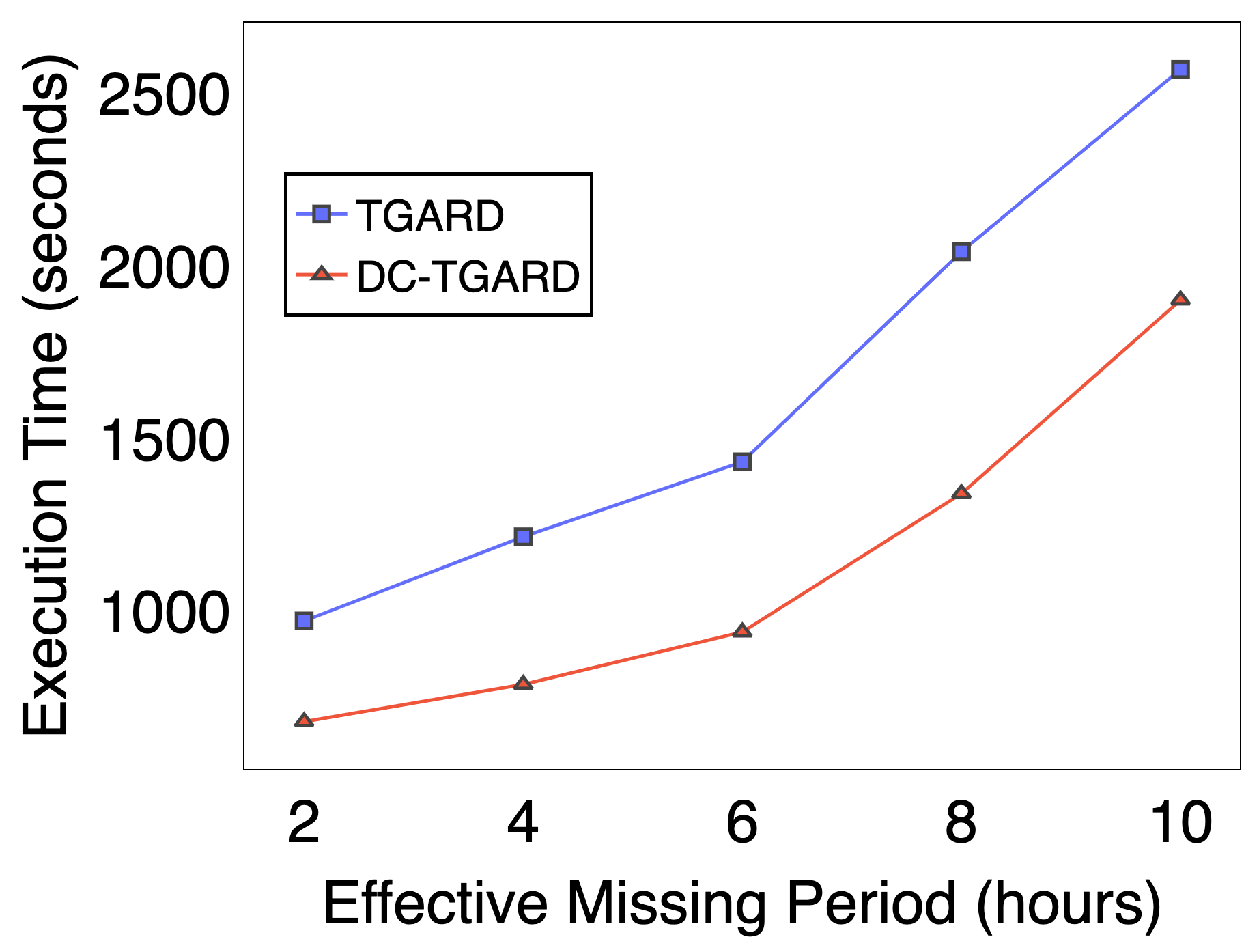}\par 
    (b) Effective Missing Period\\
\end{multicols}
\begin{multicols}{2}
    \includegraphics[width=\linewidth]{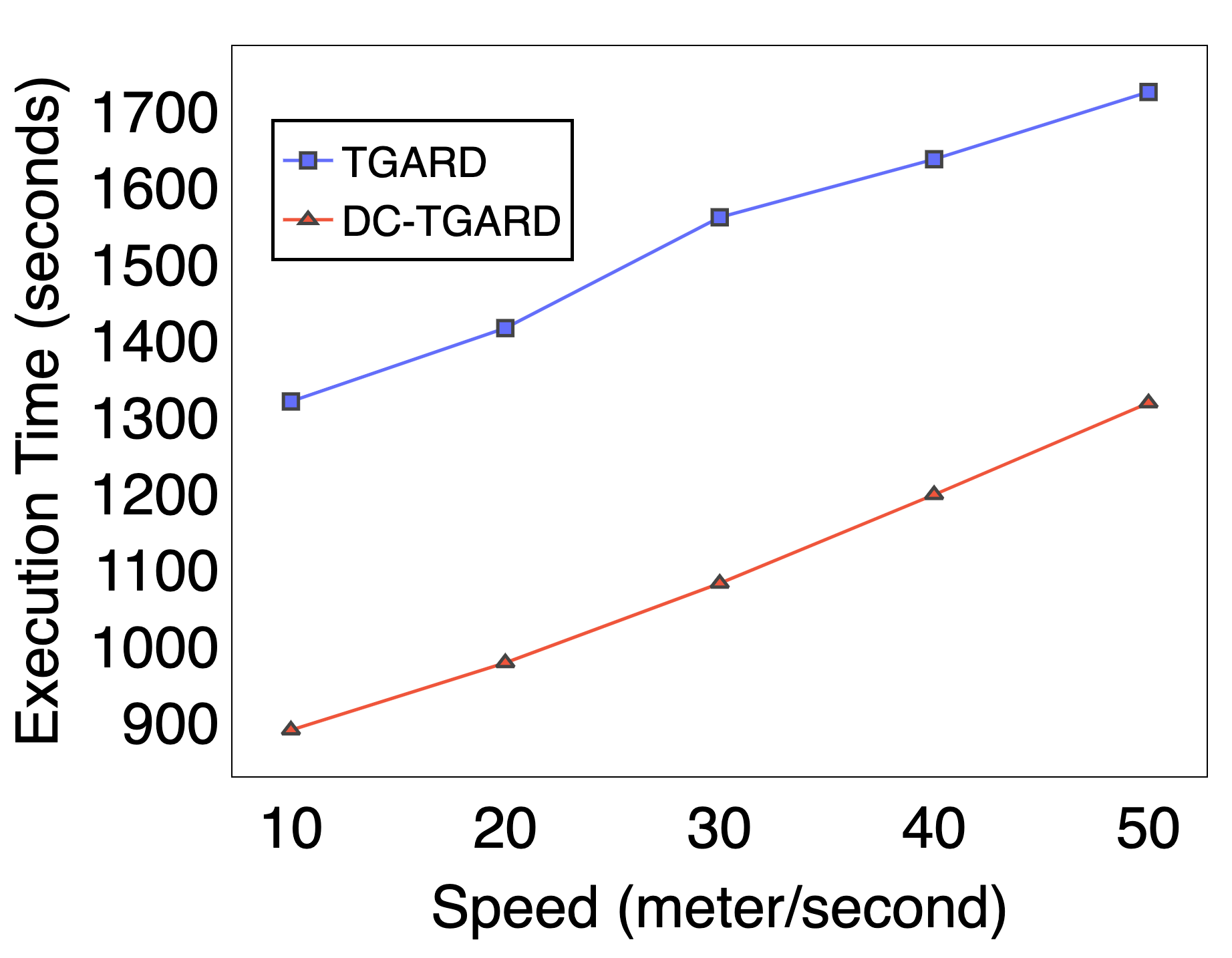}\par
    (c) Speed\\
    \includegraphics[width=\linewidth]{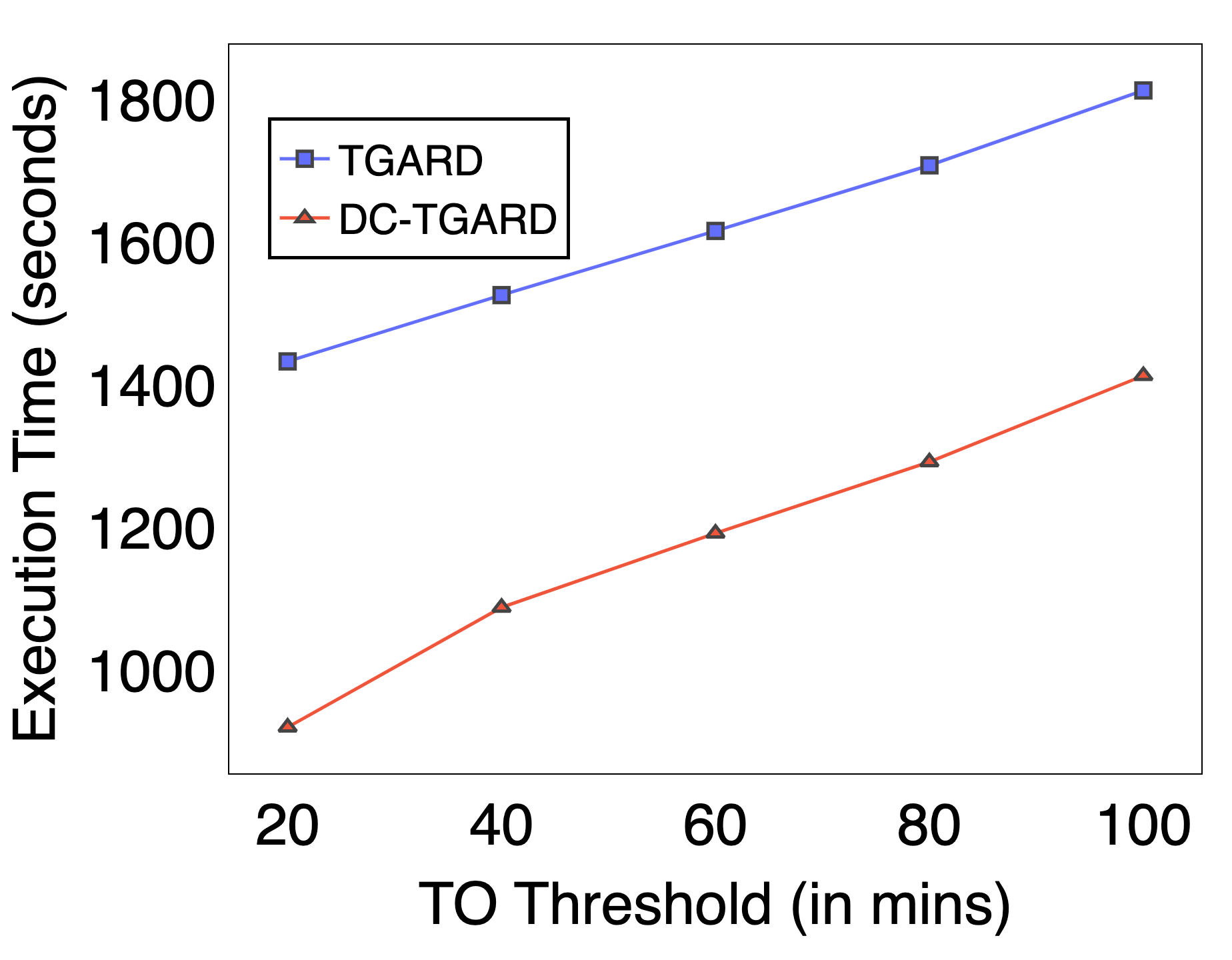}\par
    (d) TO Threshold (TO)\\
\end{multicols}
\begin{multicols}{2}
    \includegraphics[width=\linewidth]{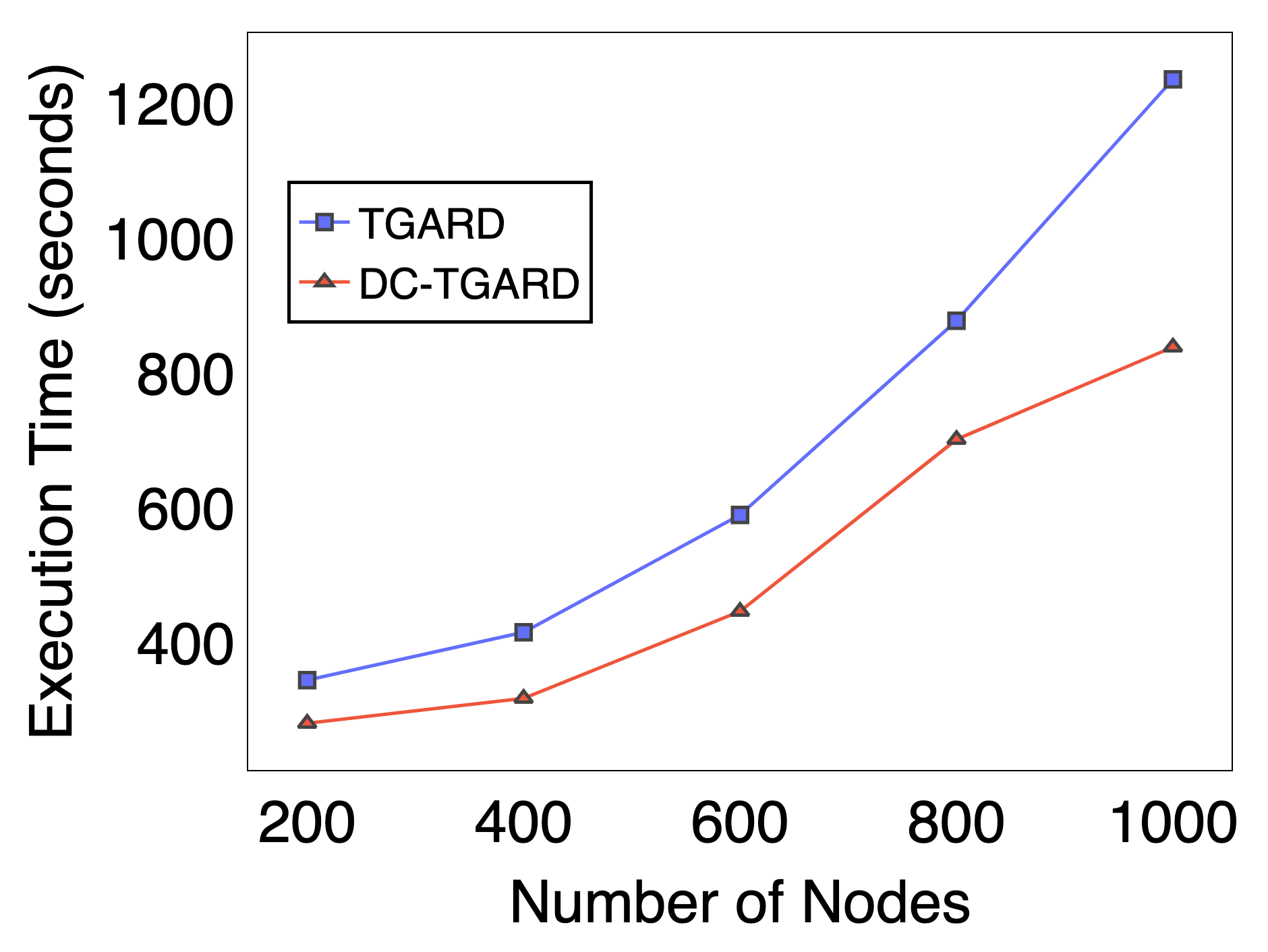}\par
    (e) Number of Nodes (Low Density)\\
    \includegraphics[width=\linewidth]{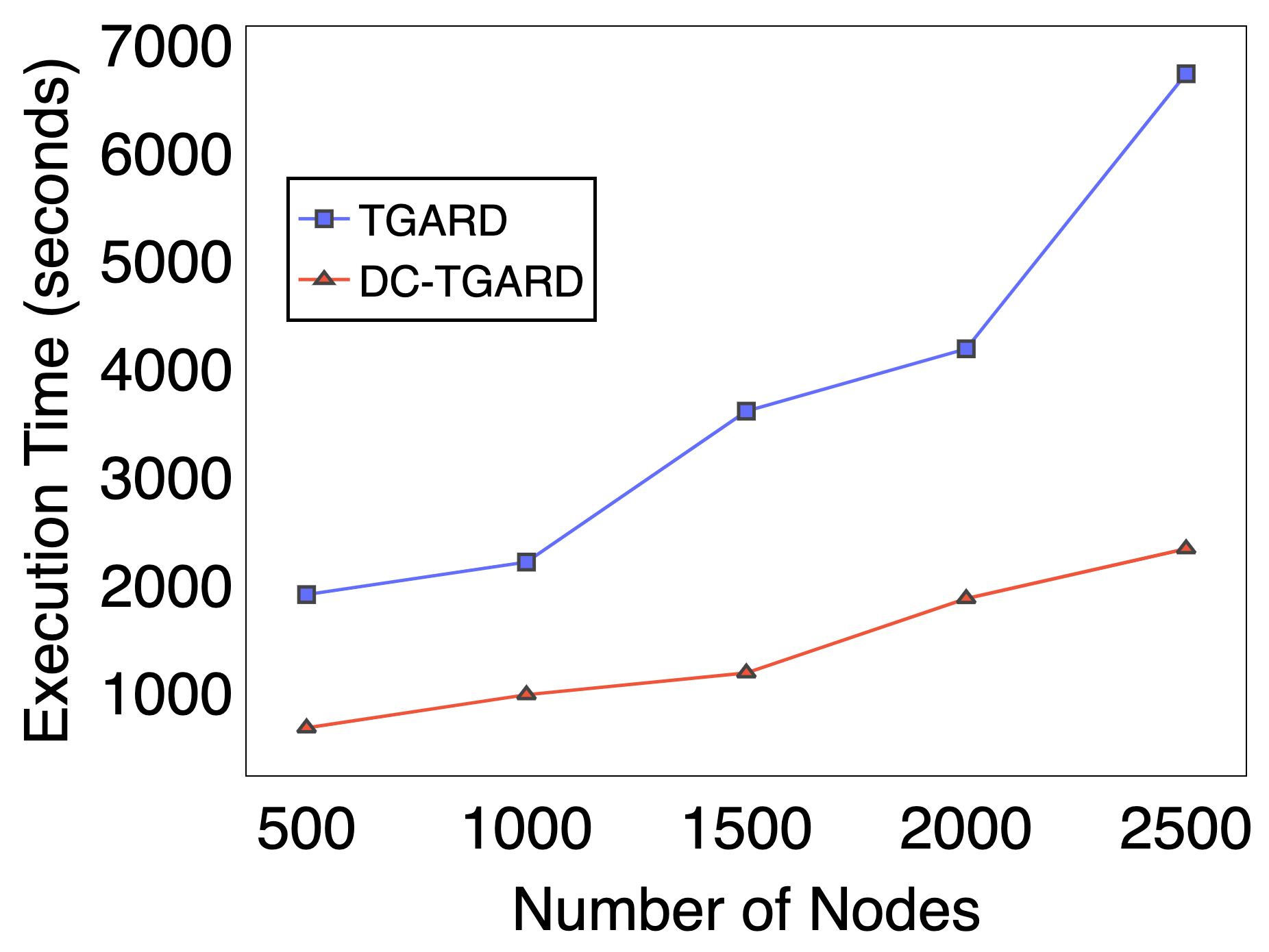}\par
    (f) Number of Nodes (High Density)\\
\end{multicols}
\caption{DC-TGARD is more efficient than TGARD Methods under different parameters}
\end{figure}
\section{Other Related Work}
\label{section:Related_work}
The literature of trajectory mining includes a broad overview of movement patterns \cite{dodge2008towards} and a taxonomy of spatial mining methods used in various application domains \cite{zheng2015trajectory}. In trajectory data management, specific frameworks have analyzed trajectory gaps via indexing methods (e.g., hierarchical trees \cite{kollios1999indexing,chen2008st2b}, grids \cite{patel2004stripes} but they are not designed to detect movement patterns. Other works have modelled regions of uncertainty \cite{trajcevski2004managing,trajcevski2003probabilistic} via snapshot models etc. More realistic solutions are based on geometric models such as cylindrical \cite{trajcevski2010uncertain} and space-time prism models \cite{miller1991modelling, kuijpers2011analytic,kim2003space} that construct an areal interpolation of the gaps using coordinates and a maximum speed of the objects. More recently, a kinetic prism \cite{kuijpers2017kinetic} approach showed improved estimation by considering other physical parameters (e.g., acceleration). However, no studies have addressed rendezvous behavior patterns within trajectory gaps.

In spatial network research, many methods consider \cite{greenfeld2002matching,lou2009map,yuan2010interactive} uncertainty in trajectory gaps by a map matching the potential routes taken by a moving object in a road network. Many studies are based on deterministic \cite{zheng2011probabilistic,zheng2012reducing} and probabilistic \cite{krumm2006predestination,krumm2013destination} methods that consider historic trajectories to find potential routes in a road network topology. However, none of these works consider rendezvous behavioral patterns in spatial networks. Recent work \cite{uddin2017assembly}, studies rendezvous detection in the \textit{static-spatial network} via space-time prisms based on objects and duration using type-2 uncertainty (geo-ellipse) resulting rather loose bounds. In this work, we are identifying rendezvous nodes via a time-slicing method\cite{sharma2022analyzing} which tighten the bounds of rendezvous areas. We further consider dynamic edge weights \cite{zhao2008algorithm,demiryurek2011online,ding2008finding} for better solution quality.

\section{Conclusion and Future Work}
\label{section:conclusion_future_work}
We study the problem of identifying a set of possible rendezvous locations within a trajectory gap in a given spatial network. We theoretically study a time slicing model which provides tighter bounds as compared to traditional space-time prism models. We also proposed a TGARD algorithm for effectively finding rendezvous nodes while filtering non-reachable nodes in the spatial networks to improve solution quality. In addition, we refined our baseline approach with a Dual Convergence TGARD (DC-TGARD) algorithm using bi-directional pruning to further improve computational efficiency. Experimental results on both synthetic and real-world datasets shows that DC-TGARD is faster than TGARD and has better solution quality than other related work based methods.

\textbf{Future Work:} In the future, we plan to investigate the adjustments of the proposed algorithm needed to handle possibly negative travel costs.
We also plan to examine whether parallelism can be used to further improve query processing. Finally, we intend to develop a more accurate prism model to calibrate realistic physics-based parameters (e.g., acceleration).
\section*{Acknowledgments}

This research is funded by an academic grant from the National Geospatial-Intelligence Agency (Award No. HM0476-20-1-0009, Project Title: Identifying Aberration Patterns in Multi-attribute Trajectory Data with Gaps). Approved for public release, 22-536. We also want to thank Kim Koffolt and the spatial computing re- search group for their helpful comments and refinements.

\bibliographystyle{ACM-Reference-Format}
\bibliography{manuscript}

\end{document}